\documentclass[aps,prx,reprint,amsmath,amssymb,superscriptaddress,showkeys]{revtex4-1}

\usepackage{graphicx}
\usepackage{amsmath}
\usepackage{amsfonts}
\usepackage{amsmath}
\usepackage{amssymb}
\usepackage{soul}
\usepackage{url}
\usepackage{bm}
\usepackage{amsthm}
\usepackage{mathtools}
\usepackage{dsfont}
\usepackage{stmaryrd}
\usepackage{adjustbox}
\usepackage{makecell}
\usepackage{xspace}
\usepackage{siunitx}
\usepackage{enumitem}
\usepackage{epigraph}
\usepackage{cancel}
\usepackage{lipsum} 
\usepackage{accents}
\usepackage[dvipsnames]{xcolor}
\usepackage{hyperref}
\usepackage{cleveref}
\usepackage{tikz-cd}
\usepackage{amsfonts}
\usepackage{hyperref}
\usepackage{cleveref}
\usepackage{enumitem}
\usepackage{ctable}

\def\compiletikz{0}

\if1\compiletikz
\input{tikzsettings.tex}
\fi

\newcommand\myshade{85}
\colorlet{mylinkcolor}{BrickRed}
\colorlet{mycitecolor}{NavyBlue}
\colorlet{myurlcolor}{Aquamarine}

\hypersetup{
  linkcolor  = mylinkcolor!\myshade!black,
  citecolor  = mycitecolor!\myshade!black,
  urlcolor   = myurlcolor!\myshade!black,
  colorlinks = true,
}

\newcommand{\past}[2]{\cev{\bm #1}_{#2}}
\newcommand{\fut}[2]{\vec{\bm #1}_{#2}}
\newcommand{\finfut}[3]{\vec{\bm #1}^{#3}_{#2}}
\newcommand{\Mod}[1]{\ (\mathrm{mod}\ #1)}

\let\oldparagraph\paragraph
\renewcommand\paragraph[1]{\vskip0.5cm\oldparagraph{#1}}

\makeatletter
\DeclareRobustCommand{\cev}[1]{%
  \mathpalette\do@cev{#1}%
}
\newcommand{\do@cev}[2]{%
  \fix@cev{#1}{+}%
  \reflectbox{$\m@th#1\vec{\reflectbox{$\fix@cev{#1}{-}\m@th#1#2\fix@cev{#1}{+}$}}$}%
  \fix@cev{#1}{-}%
}
\newcommand{\fix@cev}[2]{%
  \ifx#1\displaystyle
    \mkern#23mu
  \else
    \ifx#1\textstyle
      \mkern#23mu
    \else
      \ifx#1\scriptstyle
        \mkern#22mu
      \else
        \mkern#22mu
      \fi
    \fi
  \fi
}
\makeatother

\newcommand{\splitatcommas}[1]{%
\begingroup
\begingroup\lccode`~=`, \lowercase{\endgroup
    \edef~{\mathchar\the\mathcode`, \penalty0 \noexpand\hspace{0pt plus 1em}}%
    }\mathcode`,="8000 #1%
    \endgroup
}

\definecolor{darkgreen}{RGB}{50,200,8}

\graphicspath{ {./figs/} }

\newtheorem{definition}{Definition}
\newtheorem{theorem}{Theorem}
\newtheorem{lemma}{Lemma}
\newtheorem{proposition}{Proposition}

\newtheorem{example}{Counterexample}

\newtheorem{ex}{Example}

\newcommand{\X}{\mathcal{X}}

\newcommand*{\diff}{\mathop{\kern0pt\mathrm{d}}\!{}}

\begin{document}
\title{Software in the natural world: 
A computational approach to hierarchical emergence}

\author{Fernando E. Rosas} 
\email{email: f.rosas@sussex.ac.uk}
\affiliation{Department of Informatics, University of Sussex}
\affiliation{Sussex Centre for Consciousness Science and Sussex AI, University of Sussex}
\affiliation{Center for Psychedelic Research and Centre for Complexity Science, Department of Brain Science, Imperial College London}
\affiliation{Center for Eudaimonia and Human Flourishing, University of Oxford}

\author{Bernhard C. Geiger}
\affiliation{Know-Center GmbH, Graz, Austria}
\affiliation{Signal Processing and Speech Communication Laboratory, Graz University of Technology, Graz, Austria}

\author{Andrea~I~Luppi}
\affiliation{Montreal Neurological Institute, McGill University}

\author{Anil~K.~Seth} 
\affiliation{Department of Informatics, University of Sussex}
\affiliation{Sussex Centre for Consciousness Science and Sussex AI, University of Sussex}

\author{Daniel Polani}
\affiliation{Department of Computer Science, University of Hertfordshire, Hatfield, UK}

\author{Michael Gastpar}
\affiliation{School of Computer and Communication Sciences, EPFL, Lausanne, Switzerland} 

\author{Pedro A.M. Mediano}
\affiliation{Department of Computing, Imperial College London}
\affiliation{Division of Psychology and Language Sciences, University College London}

\begin{abstract}
\noindent
Understanding the functional architecture of complex systems is crucial to illuminate their inner workings and enable effective methods for their prediction and control. 
Recent advances have introduced tools to characterise emergent macroscopic levels; however, while these approaches are successful in identifying \emph{when} emergence takes place, they are limited in the extent they can determine \emph{how} it does. 
Here we address this important limitation by developing a computational approach to emergence, which characterises macroscopic processes in terms of their computational capabilities. 
Concretely, we articulate a view on emergence based on how software works, which is rooted on a mathematical formalisation of how macroscopic processes can express self-contained informational, interventional, and computational properties. 
This framework reveals a hierarchy of nested self-contained processes that determines what computations take place at what level, which in turn delineates the functional architecture of a complex system. 
This approach is illustrated on paradigmatic models from the statistical physics and computational neuroscience literature, which are shown to exhibit macroscopic processes that are akin to software in human-engineered systems. 
Overall, this framework enables a deeper understanding of the multi-level structure of complex systems, revealing specific ways in which they can be efficiently simulated, predicted, and controlled.
\end{abstract}

\maketitle

\section{Introduction}

Complex systems --- such as the global economy, the global weather, and the human brain --- are composed of many elements whose intricate dynamical interactions give rise to distinct phenomena at various spatio-temporal scales~\cite{waldrop1993complexity,jensen2022complexity}. 
For example, the interaction between the synaptic activity of individual neurons can give rise to mesoscopic oscillatory patterns involving entire brain regions~\cite{deco2008dynamic}, which in turn exhibit distinct coupling patterns at a macroscopic, whole-brain level that are implicated in distinct cognitive functions~\cite{fox2005human}. 
Crucially, phenomena at certain scales may not be relevant for identifying what is relevant for another; for instance, particle physics may not be relevant for neuroscience despite neurons being made of particles. 

In order to successfully understand, predict, and control such complex systems, it is crucial to be able to characterise their \emph{functional architecture} --- i.e., to identify the relevant macroscopic spatio-temporal levels and delineate the activity that `makes a difference' for each of them. Unfortunately, such characterisation is often done in ad-hoc ways, as finding rigorous, principled, and general methods is highly non-trivial~\cite{jensen2022complexity,pessoa2023entangled}. 
Importantly, this is not just a problem of lack of data or insufficient system specification: as the opacity of deep learning models dramatically illustrates~\cite{lipton2018mythos,zhang2021survey}, even a full description of the microscopic components and their interactions may not facilitate the identification of when and where emergent phenomena may take place.

Providing a rigorous characterisation of the functional architecture of multi-level complex systems remains an important, yet unsolved challenge. 
A promising line of investigation has developed general methods to identify and characterise macroscopic levels with emergent properties~\cite{seth2010measuring,hoel2013quantifying,rosas2020reconciling,varley2021emergence,barnett2021dynamical}. 
These works provide new information-theoretic tools for researchers to rigorously frame conjectures about emergence on a wide range of dynamical processes, and statistical tools to test these conjectures on data. 
Despite being only recently developed, the practical efficacy of these approaches has already been demonstrated in a number of empirical investigations~\cite{klein2020emergence,klein2021evolution,luppi2022synergistic,luppi2023reduced,proca2022synergistic,sas2024synch}. 

However, as promising as these developments are in providing effective methods to determine \emph{when} emergence takes place, they are limited in the degree they can specify \emph{how} this happens. For example, existing methods might be able to determine whether the dynamics of a recurrent neural network have emergent character via the calculation of specific information-theoretic quantities. However, these approaches could not specify which exactly are the underlying dynamical properties facilitating this, or how these properties are distributed across the various spatio-temporal scales of the system. Answering those questions is critical to deepen our understanding of how those systems work, and under what conditions they may break down.

Here we address these open questions by investigating emergence from a computational perspective, complementing  
information-theoretic methods with principles of theoretical computer science --- more specifically, automata theory~\cite{minsky1967computation} and computational mechanics~\cite{shalizi2001causal}. 
Concretely, we investigate emergence by operationalising three ways in which a macroscopic process can be self-contained: informational, causal, and computational closure --- which correspond to optimal prediction of spontaneous dynamics, exhaustive capability to
perform interventions, and containment of computational processing (respectively). By developing these notions within a unified mathematical framework, this approach establishes formal foundations to reason about the emergence of self-contained \emph{effective macroscopic theories} within microscopic stochastic dynamical processes, which are akin to software in human-engineered systems in their ability to that carry out computations independently of the substrate they are implemented in.

The results facilitated by our framework specify the conditions under which systems give rise to informational, causal, and computational closed macroscopic levels, and clarify the relationships between these properties. 
Furthermore, our framework reveals how computationally closed levels are hierarchically organised into a lattice of nested computational structures ordered by coarse-graining relationships. This lattice constitutes a blueprint for the functional architecture of a complex system, which describes what computations take place at what level. 
Additionally, our results establish links with the theory of lumpable time series~\cite{Kemeny_FMC}, which open the door for efficient algorithms to estimate these properties in real-world applications. 
The power and scope of this approach is illustrated in various case studies, including diffusion processes, spin models, random walks on networks, agent-based models, and neural systems that model processes of memory consolidation and retrieval. 
Our framework deepens our understanding of the multi-level dynamics of these systems, revealing specific ways in which they can be efficiently predicted and controlled. Moreover, its generality and broad applicability opens the way for a wide range of future theoretical developments and practical applications.

The rest of the paper is structured as follows. First Section~\ref{sec:preliminaries} provides background information, and then Section~\ref{sec:our_proposal} gives an intuitive overview of the proposed framework. After this, Section~\ref{sec:examples} presents a number of case studies to showcase the proposed framework. Finally, Section~\ref{sec:discussion} outlines the main implications of our findings, and discusses related work. A complete exposition of the formal theory supporting our framework is provided in the Appendix.

\section{Preliminaries}
\label{sec:preliminaries}

\subsection{The `software-ness' of software}
\label{sec:softwareness}

Here we investigate systems whose macroscopic levels posses a certain degree of causal `self-containment' with respect to their microscopic instantiation, which we describe as \emph{emergent}. Under this view, phenomena taking place at an emergent macroscopic level depend solely on other phenomena at the same level, without being causally affected by more fine-grained distinctions observed at microscopic levels. To illustrate these ideas, consider a simple scenario where a computer is running a script that executes the command \texttt{c=a+b}. Crucially, the value assigned to the variable \texttt{c} depends on the values assigned to \texttt{a} and \texttt{b}, but does not depend on peculiarities of the physical substrate over which the program is running (e.g.\ the spin of the electrons that make the transistors, etc). This relative autonomy from its substrate is a key feature of what \textit{software} is: 
software is always `running on top' of hardware while having some `life of its own' --- following its own rules irrespective of many details of the substrate over which it is instantiated.

Following this line of thinking, one can say that \textit{a system is running software} when it has a macroscopic scale that is \textit{causally closed} with respect to lower scales. 
Technically speaking, causal closure takes place when interventions at a set of macroscopic (i.e.\ aggregated, coarse-grained) variables are sufficient to guarantee specific outcomes on other macroscopic variables. 
Consequently, `programs' can be understood as specific (sets of) relations between macroscopic variables that do not depend on specific microscale instantiations --- in fact, one can imagine many different microscale instantiations with precisely equivalent macroscopic behaviour. 
In the case of the computer script considered above, the value assigned to the variable \texttt{c} only depends on the value previously assigned to \texttt{a} and \texttt{b}, and not on details of the physical instantiation of those variables in terms of transistors and other hardware aspects. 
Hence, causal closure implies that events that happen at a macroscale are `shielded' from how they are instantiated at a microscopic level --- i.e.\  from their implementation.

The distinction of software/hardware has been instrumental in the development of our modern technological world~\cite{patterson2013computer}. Distinguishing software from hardware is what allows designers to create programs that can be implemented over an variety of different physical substrates. 
From a philosophical perspective, the notions of software and causal closure are related to the principle of `multiple realisability'~\cite{bechtel1999multiple,polger2016multiple}, which states that there being many different microscopic ways to realise a specific macroscopic function, and the idea of `substrate independence'~\cite{bostrom2003we}, which highlights that some functions may be agnostic with respect to the substrate over which they are implemented (e.g.\ carbon vs silicon). 
From a scientific perspective, these notions are related to the idea of \textit{minimal model explanations}, which aims to explain a phenomenon via a minimal set of macroscopic features that remain stable after perturbations to their microscopic details~\cite{batterman2014minimal}.

While the idea of software was conceived in the context of device design, an interesting question arises as to whether natural, non-engineered systems could be usefully described in terms of hardware-software distinctions. 
For example, under what conditions is it reasonable to say that different natural systems (e.g.\ different living organisms of the same species) are running the same software? 
Also, would it be reasonable to think of the human brain as running software over neural hardware~\cite{block1995mind,bostrom2003we}\footnote{We remark that a simplistic interpretation of the brain/mind dichotomy in terms of a hardware/software distinction has most likely long outlived its usefulness~\cite{cobb2020idea}. Nonetheless, a computational perspective might still usefully be applied to brains to identify emergent functional architectures.}?

Arguments can be made, based on evolutionary and thermodynamic considerations, for why and to what degree biological systems may display software-like, macroscopic levels that are causally closed. Living organisms perform actions (i.e.\ interventions over their environments and themselves) in order to guarantee their survival. However, some interventions are more feasible to implement than others --- e.g., while it is almost impossible to intervene on the momentum of a particular molecule in a glass of water, it is much easier to increase its mean temperature together with all other molecules in the glass. 
Causally closed systems are those for which macroscopic interventions are efficacious in controlling their macroscopic scale. 
Hence, one can say that systems are \emph{as if} running software to the extent they are controllable at a particular scale at a reasonable thermodynamic cost. 
Building on similar arguments, causal closure has often been regarded by theoretical biologists as a necessary --- albeit not sufficient --- condition for the subsistence of living systems~\cite{varela1979principles,chandler2000emergent,pattee2012cell}.

In this paper we argue for a pragmatic approach to these issues, as our motivation is the development of analytical methods to turn these ideas into empirical questions that can be investigated quantitatively. While causal closure is certainly a necessary condition for software-like processes to exist, in this paper we explore the consequences of taking this condition to be also sufficient. 
Indeed, our framework reveals that causal closure is closely connected with a notion of \emph{computational closure}, which opens --- as our theory shows --- fertile avenues to analyse multi-level physical processes with the tools of theoretical computer science.

\subsection{From causality to computation}
\label{sec:two_faces_emachines}

\begin{figure}[th]
  \centering
  \if1\compiletikz
  \includetikz{tikz/}{emachine_intro}
  \else
    \includegraphics[width=\columnwidth]{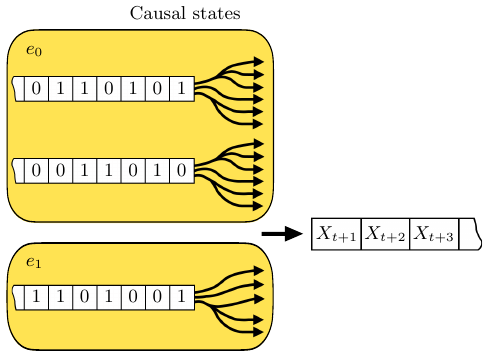}
  \fi \caption{\textbf{Illustration of causal states}. Causal states are sets of  
  of trajectories which bear equal predictions for the future evolution
of the system, as defined by the equivalence relationship in Eq.~\eqref{ref:e_machine}.}
  \label{fig:causal_states}
\end{figure} 

Above we described the nature of software from the perspective of interventions and causal closure, and suggested how it can be used for investigating natural systems. However, while such approach would allow us to determine if a system is running software or not in terms of causal relations, it does not facilitate a direct way to describe this in terms of \emph{computations}. 
One of the central aims of this paper is to illuminate the computational implications of causal closure.

To attain this, our approach takes inspiration and builds upon the rich literature of computational mechanics --- a framework for the study of stochastic dynamical processes from a computational perspective~\cite{crutchfield1994calculi,shalizi2001computational,shalizi2001causal,rupe2023principles}. Specifically, computational mechanics studies patterns and statistical regularities observed in time series data by asking quantitative questions such as: how much historical information does the system store, where is that information stored, and how is it processed to generate future behaviour~\cite{amslaurea15649,crutchfield2017origins,grassberger1986toward}? The answers to those questions, according to computational mechanics, can be found in the states and transitions of so-called `$\epsilon$-machines'~\cite{grassberger1986toward,crutchfield1989inferring,crutchfield1994calculi}.

\begin{figure*}
  \centering
  \if1\compiletikz
  \includetikz{tikz/}{diagram}
  \else
    \includegraphics{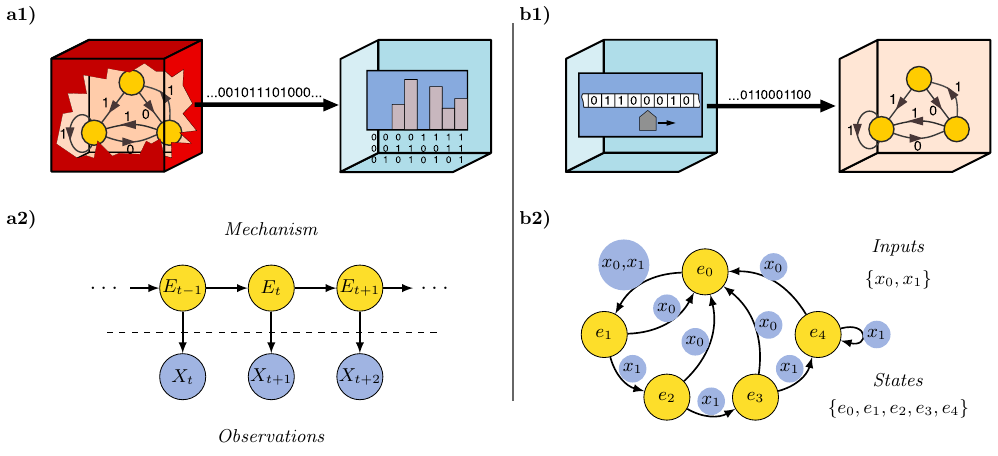}
  \fi 
  \caption{\textbf{The two faces of $\epsilon$-machines}.
  Illustration of the dual interpretation of $\epsilon$-machines that establish a bridge between causality and computation. 
  \textbf{a) Causal face}: View of $\epsilon$-machines as the effective mechanism driving the system, acting `behind the scenes' to generate observable data (a1). Technically, this corresponds to interpreting it as a hidden Markov process 
  --- i.e., dynamics that take place on variables $E_t$ on a latent state-space, while generating the observable data
$X_t$ (a2). \textbf{b) Computational face.} Alternative view of $\epsilon$-machines as discrete automata, where the data corresponds to inputs given by a user driving the system between different states (b1). Technically, this corresponds to seeing it as a discrete automata with states $e_k$, whose deterministic transitions are governed by the input data $x_i$ (b2). Note that (a1) focuses on variables (e.g.\ $X_t,E_t$), while (b2) portraits the states that those variables can take (e.g.\ $x_0,e_0$). Fig. (a1) is adapted from Ref.~\cite{crutchfield2001regularities}.}
\label{fig:SS_extended}
\end{figure*} 

Computational mechanics introduces $\epsilon$-machines as optimal (i.e.\  complete and minimimal) representations of a computational `engine' that can generate the patterns observed in data. 
Such representation can be built by grouping past trajectories according to their forecasted futures into so-called \textit{causal states} (see Figure~\ref{fig:causal_states}). 
More precisely, the causal states of a time series $\bm X = \{X_t\}_{t\in\mathbb{N}}$ 
are the equivalence classes of past trajectories $\past{x}{t} \coloneqq (\dots, x_{t-1}, x_{t})$ that are established by the following equivalence relationship:
\begin{equation}
\label{ref:e_machine}
    \past{x}{t} \equiv_\epsilon \past{x}{t}' \quad \text{iff} \quad p(\finfut{x}{t+1}{L} | \past{x}{t}) = p(\finfut{x}{t+1}{L} | \past{x}{t}') \quad \forall \finfut{x}{t+1}{L}, L\in\mathbb{N},
\end{equation}
where $\finfut{x}{t+1}{L}=(x_{t+1},\dots,x_{t+1+L})$ and $p(a|b)$ denotes the probability of $a$ given $b$~\footnote{Through the article we use uppercase letters ($X,Y$) to denote random variables and lowercase letters ($x,y$) to denote their realisations.}. 
Hence, the causal states are given by a mapping $\epsilon$ which coarse-grains past trajectories $\past{X}{t}$, thereby generating a new time series of causal states $\bm E=\{E_t\}_{t\in\mathbb{N}}$ with  $E_t=\epsilon(\past{X}{t})$. Formally, the $\epsilon$-machine of $\bm X$ is  defined as the pair $(\epsilon,T_{e,e'}^x)$, where $\epsilon$ is the above mapping to causal states and $T_{e,e'}^x$ is a collection of transition probabilities of the form
\begin{equation}
    T_{e,e'}^x = \mathbb{P}\{ E_t=e', X_t=x | E_{t-1}=e\}~.
\end{equation}
Hence, the causal states correspond to a kind of information bottleneck~\cite{shalizi2001computational}: $\epsilon$ is the coarsest coarse-graining of past trajectories $\past{x}{t}$ that retains full predictive power over future variables (i.e.\  $I(\past{X}{t}; \finfut{X}{t+1}{L}) = I(E_t; \finfut{X}{t+1}{L})$ for all $L\in\mathbb{N}$). 
Additionally, the causal states have Markov dynamics, and by `Markovianising' the process this forces the effective states of the process to become explicit. 

In summary, the $\epsilon$-machine provides the simplest  computational mechanism capable of giving rise to the system's dynamics. Hence, the $\epsilon$-machine of a process can be regarded as its effective theory --- akin to its equations of motion~\cite{crutchfield2017origins}. It is important, nonetheless, to note that the causal states provide counterfactual guarantees only if the system at hand is fully observed, or equivalently, if the conditional probabilities that describe $\bm X$ are respected by interventions~\cite{pearl2009causality}. In the case of partially observed scenarios, causal states ought to be understood in the Granger sense, i.e.\ as states of maximal non-mediated predictive ability~\cite{bressler2011wiener}, and as such, the best possible attempt to identify the causal drivers given the available knowledge.

A key contribution of computational mechanics is establishing a bridge between causality and computation, which is rooted in two alternative perspectives that can be used to interpret what $\epsilon$-machines are (see Figure~\ref{fig:SS_extended}):
\begin{itemize}
\item[\textbf{a)}] From a causal perspective, $\epsilon$-machines can be conceived as Markov processes taking place at a hidden, latent space~\cite{ephraim2002hidden}, corresponding to the underlying mechanism at work `under the hood.' From this view, the system is driven by the dynamics of causal states, and the data that one measures is a consequence of the transitions between them. 
\item[\textbf{b)}] From a computational perspective, $\epsilon$-machines can be thought of as deterministic automata --- i.e.\ (finite or infinite) state machines whose transitions are fully determined by an input sequence~\cite{hopcroft2001introduction}. Deterministic automata are a paradigmatic example of computational systems (typically used to drive vending machines, elevators, and traffic lights), where a string of inputs is used to drive the system through a sequence of internal states~\footnote{For an introduction to automata and their interpretation as computational systems, please see Refs.~\cite{minsky1967computation,hopcroft2001introduction}.}. 
Under this view, the input driving the automata is the observed data, and the dynamics over causal states is the result of those.     
\end{itemize}
These two quite distinct views are both consistent with what a
$\epsilon$-machine formally is, and their duality opens a principled link between principles of causality and computation~\footnote{This equivalence between hidden Markov models and discrete automatas is not general, but only holds because the Markov dynamics of $\epsilon$-machines have the so-called \emph{unifilar} property: transitions between causal states (e.g.\ from $E_t$ to $E_{t+1}$) are deterministic after the next symbol (in this case $X_{t+1}$ has been emitted.}.

\section{A computational approach to emergence}
\label{sec:our_proposal}

After setting the background and clarifying foundational ideas, this section presents the ideas at the core of our proposed framework in an intuitive manner. The exposition provides links to formal definitions and theorems in the \hyperref[sec:theory]{Appendix}, where the theory is presented with full technical detail. The ideas discussed here are then illustrated by a number of case studies in Section~\ref{sec:examples}.

\subsection{What drives the macro? A story of two machines}
\label{sec:what_caused_what}

Our approach starts by combining computational mechanics~\cite{shalizi2001computational,shalizi2001causal} and an interventionist view on causation~\cite{pearl2018book} to investigate the causality of macroscopic phenomena. 
Concretely, let us consider a dynamical process $X_t$ and a coarse-grained, macroscopic process $Z_t = f(X_t)$ derived from it, and ask: what are the causes of events in $Z_t$?

We address this question by building on the principle that a cause of an effect is singled out by the minimal set of distinctions that make a difference for that level. 
Concretely, by taking inspiration from $\epsilon$-machines (see Section~\ref{sec:two_faces_emachines} and Definition~\ref{def:e-machine}), our approach is to identify the minimal set of distinctions between past trajectories of $\bm X$ that lead to different potential outcomes in the future of $\bm Z$. Mathematically, we build a new computational mechanics `machine' based on the following equivalence relationship (see Definition~\ref{def:u-machine}):
\begin{equation}
    \past{x}{t} \equiv_\upsilon \past{x}{t}' \quad \text{iff} \quad p(\finfut{z}{t+1}{L} | \past{x}{t}) = p(\finfut{z}{t+1}{L} | \past{x}{t}') \quad \forall \finfut{z}{t+1}{L},L\in\mathbb{N}.
\nonumber
\end{equation}
In contrast with Eq.~\eqref{ref:e_machine}, the resulting equivalence classes are based on the prediction of the future of $\bm Z$, and not the future of $\bm X$. 
This gives rise to another, usually simpler set of distinctions in the past trajectories of $X$ (i.e.\ a different set of causal states), which we call the $\upsilon$-machine --- with upsilon the greek equivalent to the latin \textit{u} (although when combined with other letters it sounds like \textit{i}), referring to `underlying' for reasons explained below. 
The causal states of the $\upsilon$-machine are given by a mapping $\upsilon$ which coarse-grains $\past{X}{t}$, thereby generating a new time series $\bm U=\{U_t\}_{t\in\mathbb{N}}$ with  $U_t=\upsilon(\past{X}{t})$. Analogously than for the $\epsilon$-machine, our results show that the causal states of the $\upsilon$-machine to be an optimal information bottleneck --- concretely, $\upsilon$ is the coarsest coarse-graining of past trajectories $\past{x}{t}$ that retains full predictive power over the future of $\bm Z$ (see Proposition~\ref{prop:minimality_upsilon}).

\begin{figure}[th]
  \centering
  \if1\compiletikz
  \includetikz{tikz/}{e_vs_u}
  \else
    \includegraphics{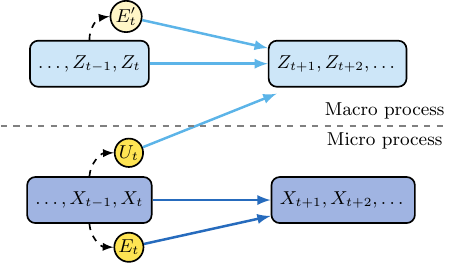}
  \fi 
  \caption{\textbf{The various machines associated with a macroscopic process}. Diagram of the relationship between the different machines associated with a macroscopic process $\bm Z$ and its corresponding microscopic process $\bm X$. The $\epsilon$-machines with causal states $E_t$ and $E'_t$ correspond to the optimal prediction of the future of $\bm X$ and $\bm Z$, respectively, using data from the same level. In contrast, the $\upsilon$-machine with causal states $U_t$ provides optimal prediction of the future of $\bm Z$ using data from $\bm X$, hence using the minimal amount of micro information for optimally predicting the future of the macro.}  
  \label{fig:e_vs_u}
\end{figure}

With this construction, each macroscopic process $Z_t=f(X_t)$ now possesses two associated machines (see Figure~\ref{fig:e_vs_u}): its $\upsilon$-machine (with causal states $U_t$) that determines the states of $\bm X$ that make a difference for the future of $\bm Z$ (Definition~\ref{def:u-machine}), 
and its $\epsilon$-machine (with causal states $E'_t$) that corresponds to the set of distinctions within the past of $\bm Z$ that make a difference for its own evolution (Definition~\ref{def:e-machine}). 
Note that the $\epsilon$-machine of $\bm Z$ is often different from the $\epsilon$-machine of the underlying microscopic process $\bm X$ (with causal states $E_t$). 
These three machines are related as follows. On the one hand, the $\upsilon$-machine of $\bm Z$ is usually more powerful than its own $\epsilon$-machine (see Lemma~\ref{lemma:u_over_e}), as it is not restricted to what can be observed from the macroscopic level, but also can access what is at the microscopic level --- as illustrated in Figure~\ref{fig:e_vs_u}. 
On the other hand, the $\upsilon$-machine of $\bm Z$ is usually weaker than the $\epsilon$-machine of $\bm X$, as the former only accounts for what is relevant for predicting $\bm Z$. 

Seen under this light, one can conclude that the $\upsilon$-machine of $\bm Z$ capture its `genuine cause' (or at least its best estimation based on the information available from the micro level), while its own $\epsilon$-machine is the best possible reconstruction of the causal states $U_t$ using only what is available at the level of $\bm Z$. That being said, please note that while the causal states $E_t$ and $U_t$ are coarse-grainings of $\past{X}{t}$ and $E'_t$ is a coarse-graining of $\past{Z}{t}$, their relationship can be highly non-trivial (see Counterexamples~\ref{ex:incompatibility_emachines} and \ref{ex:UdoesnotmaptoE}).

\subsection{Causal and information closure}

What happens if a macroscopic process $Z_t=f(X_t)$ is such that both its $\epsilon$-machines and $\upsilon$-machines are equivalent (i.e.\ if there exists a bijection between the states of the $\epsilon$- and $\upsilon$-machines)? If that is the case, then all the distinctions in the past trajectories of $\bm X$ that make a difference (i.e.\ the causal states of the macroscopic process, given by its $\upsilon$-machine) can be effectively accessed from the macroscopic level. Pragmatically, this implies that all interventions that make a difference for $\bm Z$ can be implemented at the level of $\bm Z$ without checking `microscopic details' related to how it is implemented by $\bm X$. Put simply, this means that the differences that make a difference for $Z_t$ are actually macroscopic differences. 
In this sense, it is fair to say that $Z_t$ is \emph{causing its own future} --- or more precisely, that all the causes of $Z_t$ are within its own level, and not below in a microscale. 
Building on these observations, we take this condition (i.e. the equivalence between the $\epsilon$-machine and $\upsilon$-machine of a macroscopic process $\bm Z$) as a formal definition of \emph{causal closure} (Definition~\ref{def:causal_closure}), and hence, following the discussion in Section~\ref{sec:softwareness}, we say that such are software-like processes running over their corresponding microscopic instantiations.

A somehow related notion is the idea of \textit{information closure}, which focuses on prediction instead of causation. Information closure was first introduced as a way to quantify the degree of autonomy of an organism with respect to its environment~\cite{bertschinger2006information}, and was then extended to explore the degree to which a macroscopic level depends on its underlying microscale~\cite{chang2020information}. This, in turn, has been used to investigate emergence~\cite{barnett2021dynamical}, understanding it as macroscopic processes which can optimally predict themselves. 
Concretely, a macroscopic process $\bm Z$ can be said to be informationally closed with respect to the microscopic process $\bm X$ if the following condition is satisfied (see Definition~\ref{def:infoclosure}):
\begin{equation}
I(\past{Z}{t}; \finfut{Z}{t+1}{L})
=
I(\past{X}{t}; \finfut{Z}{t+1}{L})
\quad \forall L\in\mathbb{N},
\end{equation}
where $I$ is Shannon's mutual information, which quantifies the amount of predictive information linking its arguments~\footnote{To make this point precise, prediction quality is assessed in terms of {\it expected logarithmic loss,} rather than directly via prediction error probability. See for example~\cite{bondaschi2023alphanml} for an in-depth discussion.}. 
This condition implies that knowledge about the micro-scale does not enhance the predictability of the macro-scale process above the self-predictive power of the process over itself. This makes the macro-scale `informationally sufficient' to predict itself, making it unnecessary to refer to the micro-scale~\footnote{More technically, a informationally-close macro-scale is a \textit{sufficient statistic}~\cite{casella2021statistical} with respect to the micro-scale to predict its own future.}.

What is the precise relationship between causal and informational closure?
Information closure only requires that one is able to predict the same amount of information from micro or macro, which may sound weaker than requiring that the actual structure of the underlying machines to be equivalent. 
Perhaps surprisingly, a first consequence of our framework (Theorem~\ref{thm:equivalence}) proves that information and causal closure are equivalent for any macroscopic process $Z_t=g(\past{X}{t})$. 
That said, it is important to note that causal closure (as defined here) only provides counterfactual guarantees if the system described by $\bm X$ is fully observed; otherwise, it ought to be understood in the Granger sense, i.e.\ as a best guess given the available knowledge (see the related discussion in Section~\ref{sec:two_faces_emachines}).

\subsection{Computational closure and renormalisation}

After establishing a rigorous definition of causal closure, identifying it with our intuitions of software, and clarifying its relationship with information closure, our next step is to investigate its properties in terms of computational principles. For this, we leverage the duality of $\epsilon$-machines as both state-space models and as deterministic automata (see Figure~\ref{fig:SS_extended}b) as follows.

As discussed in Section~\ref{sec:two_faces_emachines}, automata use sequences of input symbols to generate transitions between their internal states, and hence it is natural to ask: what would happen if --- due to e.g.\ noise or compression --- one only has access to coarse-grained versions of the input symbols? In general, the transitions of an automata under coarse-grained symbols are not anymore deterministic, which breaks its internal structure. However, under some circumstances it is possible to recover deterministic transitions by considering coarse-grained states --- a condition that we call \emph{computational closure}. 
More formally, computational closure applies to coarse-grainings of input symbols of an automaton for which there exists a corresponding coarse-graining of states such that the resulting transitions --- between coarse-grained states following the coarse-grained symbols --- are again deterministic (Definition~\ref{def:compu_closure}). Put simply, computational closure takes place when a new deterministic automaton can be obtained by coarse-graining both the states of the original automaton and its input alphabet (see Figure~\ref{fig:compu_closure_intuitive}).

\begin{figure}[t]
  \centering
  \if1\compiletikz
  \includetikz{tikz/}{compu_closure}
  \else
    \includegraphics{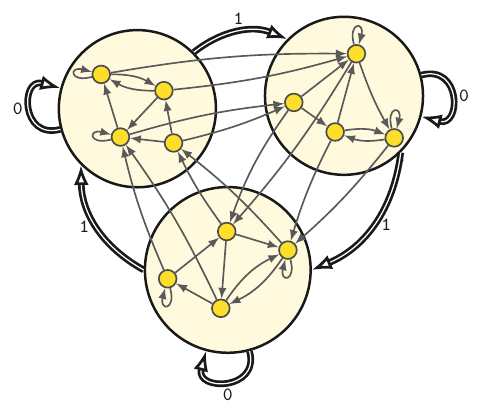}
  \fi 
  \caption{\textbf{Example of computational closure}. Illustration where micro causal states are shown as small golden nodes and macro causal states are represented as big pale-yellow nodes. 
  Transitions of micro causal states are represented as simple arrows responding to three possible inputs: two inputs denoted by \texttt{a} and \texttt{b} (not shown) trigger transitions within the same macro state, and one input denoted by \texttt{c} (not show) triggers a transition to a new macro state. The coarse-graining $f(\texttt{a})=f(\texttt{b})=\texttt{0}$ and $f(\texttt{c})=\texttt{1}$ generate deterministic dynamics for the macro states represented by double arrows, whereas \texttt{0} makes the state to remain and \texttt{1} makes a transition to the next state.}
    \label{fig:compu_closure_intuitive}
\end{figure}

\begin{figure*}
  \centering
  \if1\compiletikz
  \includetikz{tikz/}{multiscale}
  \else
    \includegraphics{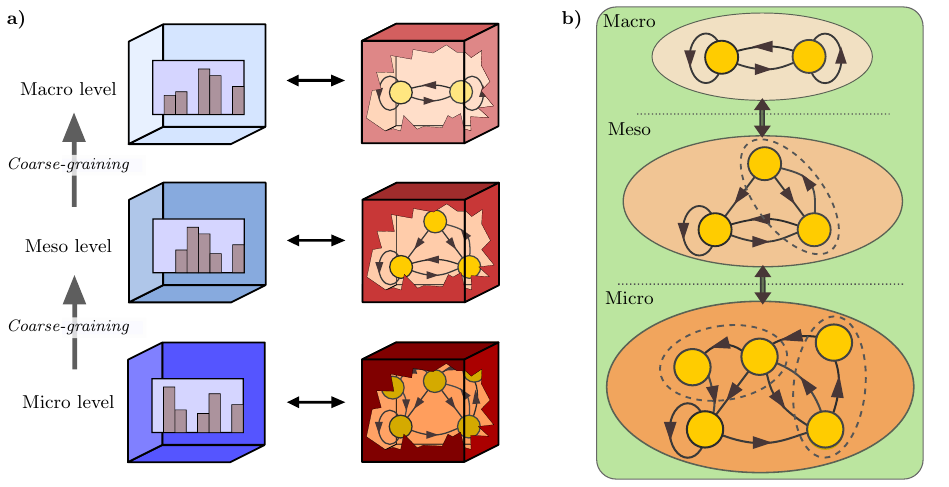}
  \fi 
  \caption{
  \textbf{Multilevel analysis via $\epsilon$-machines}. 
  \textbf{a)} Optimal automata can be built at different levels of coarse-graining of observed data. Each automaton accounts for the resulting patterns taking place at that scale. 
  \textbf{b)} If the considered levels of description are computationally closed, then the automata of higher levels are coarse-grainings of the ones of levels below. This process of coarse-graining of machines reveals the computations taking place at each of those levels.
  }
  \label{fig:multiscale}
\end{figure*}

To build some intuition about computational closure, let's consider a simple example: imagine a finite-state machine that actually consists of two independent machines, so its state is the vector of the states of each individual machine, and its input alphabet is the Cartesian product of the individual alphabets. 
It is clear that in this case different parts of the inputs satisfy the definition of computational closure,
as different parts of the input are processed separately. The definition of computational closure generalises this to cases of machines that are nested, where part of the information can be processed autonomously but another part cannot.

Computational closure have a natural characterisation for the case of $\epsilon$-machines (see Figure~\ref{fig:multiscale}). Let's consider a microscopic process $\bm X$ and a coarse-graining $\bm Z$, and denote by $E_t$ and $E'_t$ the causal states of the $\epsilon$-machines of them, respectively. Then, the $\epsilon$-machine of a macroscopic process $\bm Z$ is computationally closed if there is a coarse-graining mapping $f^*$ such that $E'_t= f^*(E_t)$. 
In effect, if that is the case, then $E'_t$ 
are coarse-grained states of the automata with states $E_t$, and $\bm Z$ provides a coarse-graining of the input symbols $\bm X$. 
Interestingly, computational closure implies that the operations of coarse-graining and calculating the $\epsilon$-machine are interchangeable, or equivalently, that the following diagram commutes:
\begin{figure}[h!]
  \centering
  \if1\compiletikz
  \includetikz{tikz/}{CommutingDiagram}
  \else
    \includegraphics{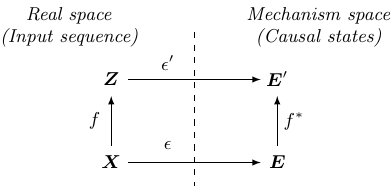}
  \fi 
\end{figure}

\noindent
where $f$ and $f^*$ are coarse-graining operations and $\epsilon$ and $\epsilon'$ are the coarse-grainings that lead to the causal states of the corresponding $\epsilon$-machines. Therefore, computationally closed levels can be said to carry out a specific portion of the computations of the whole process --- the part that remains after coarse-graining the $\epsilon$-machine of the base level.

An attractive aspect of computational closure is that it can be understood from a renormalisation theory~\cite{wilson1979problems} perspective as follows. The $\epsilon$-machines of the micro- and macroscropic processes can be seen as the `theory', or model, that best explains the statistical patterns observed at each level. 
Hence, the question we are considering can be reframed as asking how can one obtain an effective theory to explain the patterns observed in coarse-grained data. 
The commuting diagram above says that, if computational closure holds, then effective theories to explain macro levels can be found simply by coarse-graining the theory that explains the level below~\footnote{In such interpretation, the macroscopic level is the coarse-graining $\bm Z$, and its $\epsilon$-machine $E'_t$ is its corresponding theory.}. 
Therefore, macroscopic levels that are computationally closed can be thoroughly described without relying on the full microscopic theory, but using solely a simpler, coarse-grained version of it. 
Conversely, the part of the computations that take place at a computationally closed level can be extracted by an appropriate coarse-graining of the theory/$\epsilon$-machine of the microscopic level. 
Please note that this is certainly not true for arbitrary macroscopic processes, whose $\epsilon$-machines are often incompatible --- and hence incomparable --- with the one of the micro (for a minimal example, see Counterexample~\ref{ex:incompatibility_emachines}).

\subsection{Causal closure, computational closure, and lumpability}

\renewcommand{\arraystretch}{2}
\begin{table}[thbp]
\begin{center}%
\caption{Summary of main concepts}
\label{tab:sym}
\begin{tabular}{c c}%
\specialrule{.13em}{.0em}{.15em} %
Term & Explanation  \\
\specialrule{.13em}{.0em}{.15em} %
 \textit{Micro level} & Basic time-series\\[0.1cm]
 \hline
 \textit{Macro level} & \makecell{A coarse-graining of a micro process} \\[0.1cm]
 \hline
 $\epsilon$-\textit{machine} & \makecell{Model of a process built on\\ the history at the same level}\\
 \hline
 $\upsilon$-\textit{machine} & \makecell{Model of a process built on\\ the history at a level below} \\
 \hline
 \textit{Causal states} & Hidden states of $\epsilon$-  or $\upsilon$-machines \\[0.1cm]
 \hline
 \makecell{\textit{Information}\\ \textit{closure}} & \makecell{When best predictions of macro\\ can be done from the same macro}\\
 \hline
 \textit{Causal closure} & When $\epsilon$- and $\upsilon$-machines coincide \\[0.1cm]
 \hline
 \makecell{\textit{Computational}\\ \textit{closure}} & \makecell{When $\epsilon$-machine of macro is a\\ coarse-graining of $\epsilon$-machine of micro}\\
 \hline
 \textit{Lumpability} & \makecell{The ability of a Markov process to remain\\ Markovian after being coarse-grained}\\
\specialrule{.13em}{.15em}{.0em} 
\end{tabular}
\end{center}
\end{table}

At this point one may wonder: what is the relation between computational and causal closure? 
Our second main result (Theorem~\ref{teo:compu_closure}) reveals that, for macroscopic processes of the form $Z_t = f(X_t)$,
information closure implies computational closure~\footnote{Note that this implication holds only for macroscopic processes based on a spatial, but not temporal, coarse-grainings.}. 
This result may be surprising, as computational closure sounds more restrictive than information/causal closure. The reason why computational closure does not imply information/causal closure is that whereas the former is based on sufficiency, the latter also requires necessity. In effect, while for computational closure it is enough if the causal states of the macro are compatible with the causal states at the micro (i.e., induce a coarser equivalence class), causal closure also requires every relevant causal state at the micro to be accessible from the macro. Hence, a computationally closed coarse-graining may nonetheless be coarser than the partition induced by its $\upsilon$-machine, and hence fail to guarantee causal closure. For a minimal example of this, see Counterexample~\ref{ex:compcl_no_infocl}.

Theorem~\ref{teo:compu_closure} is important, as it provides a computational interpretation for causally closed systems --- i.e.,\ to macroscopic processes that can be described as software-like. Indeed, this result shows that systems that exhibit software-like macroscopic processes are characterised by having a portion of the computations that are carried out autonomously from the rest of the system, which can be efficiently controlled. 
Hence, our framework puts the arguments developed in Section~\ref{sec:softwareness} on a rigorous theoretical footing, while providing quantitative tools to investigate them in concrete, empirical data.

The relationship between causal and computational closure is further clarified by our final theoretical result, which relates these notions with \emph{weak} and \emph{strong lumpability}. 
Lumpability pertains to the study of how the memory of a stochastic process is affected after it is coarse-grained. More specifically, a memoryless (i.e.\ Markovian) process is lumpable if it has a coarse-graining that gives rise to equally memoryless dynamics --- being strongly lumpable if this only depends on the transition probabilities, and weakly lumpable if this only holds for specific initial conditions (see Definition~\ref{def:lumpability} and Ref.~\cite{Kemeny_FMC}). 
This type of inquiry leads to a question: which coarse-grainings of causal states give rise to $\epsilon$-machines of macroscopic processes of the original process? Our framework shows (Theorem~\ref{teo:lumpability}) that 
weakly lumpable coarse-grainings of causal states 
are enough to guarantee this, but the resulting levels may satisfy computational but not information/causal closure. In contrast, coarse-grainings that are strongly lumpable are guaranteed to give rise to the $\epsilon$-machine of causal/informationally closed macroscopic processes. These results imply that causal/informational closure is a more `robust' property than computational closure, as the latter can be disrupted by deviations from the stationary distribution while the former cannot.

In summary, lumpability identifies necessary and sufficient conditions for a microscopic process to have emergent macroscopic processes: a microscopic process has causally/informationally closed levels if and only if its causal states are strongly lumpable. 
By doing this, Theorem~\ref{teo:lumpability} also opens the door to the discovery of emergent levels in systems of interest. In effect, this theorem also allows us to leverage the rich literature of numerical methods that has been developed to empirically discovering lumpable partitions of a Markov chain. Thanks to this link, those methods can be readily used to empowering data-driven discovery of emergent coarse-grainings --- with the trick of applying them not in `real space' but on the space of causal states. Algorithms for finding lumpable partitions are discussed in Section~\ref{sec:lumpability}.

\subsection{The hierarchy of emergent computational levels}
\label{sec:lattice_of_emergence}

The findings discussed above have a number of important consequences. 
First, the collection of all causally closed coarse-grainings can be shown to be hierarchically organised as a lattice --- more precisely, they form a sub-lattice of the lattice of all coarse-grainings. 
Moreover, the fact that for spatial coarse-grainings information closure implies computational closure (discussed in the previous subsection) means that the collection of $\epsilon$-machines of all informationally closed spatial coarse-grainings are also hierarchically organised in \emph{another lattice}: a lattice of strongly-lumpable coarse-grainings of $\epsilon$-machines, whose finest element is the $\epsilon$-machine of the micro level and its coarsest is an $\epsilon$-machine with a single state.

\begin{figure*}[ht]
  \centering
  \if1\compiletikz
  \includetikz{tikz/}{megalattices}
  \else
    \includegraphics{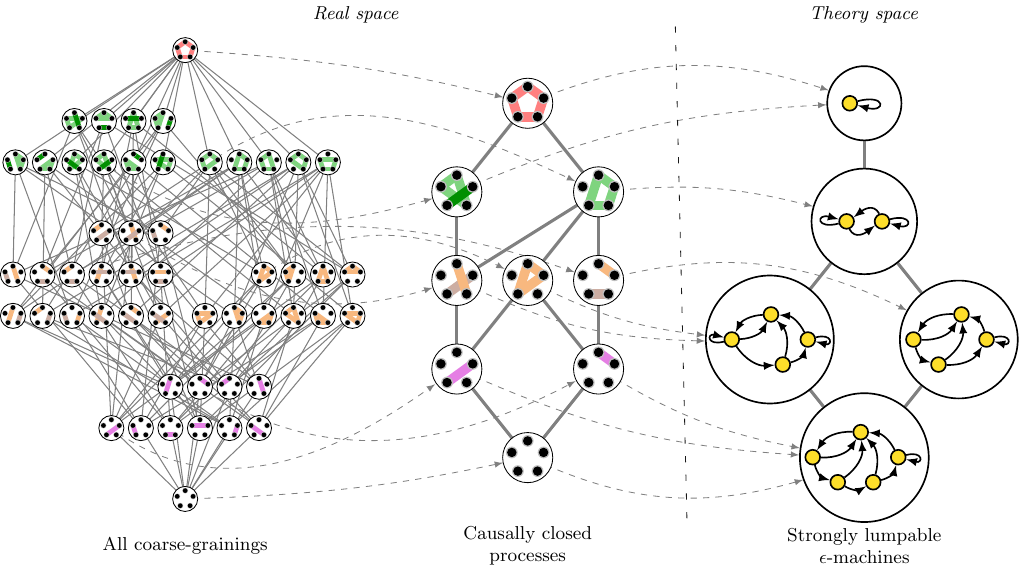}
  \fi 
  \caption{
  \textbf{The multiple hiearchies describing multi-level computations in a complex system.}
  \textit{Left}: Lattice of all possible coarse-grainings, here illustrated for the case of a process that can take five possible values. 
  \textit{Center}: Sub-lattice of only those coarse-grainings that are causally/informationally closed.
  \textit{Right}: Lattice of strongly-lumpable coarse-grainings of the $\epsilon$-machine of the microscopic level. 
  Only the last lattice provides a minimal blueprint that highlights the distinct computational processes, and distinguishes which computations take place at what level.
  }
  \label{fig:megalattice}
\end{figure*}

This sequence of three compatible hierarchies is important, as the mapping from one lattice to the next provide a more refined and simpler view of the multi-level structure of the system of interest (see Figure~\ref{fig:megalattice}). 
Indeed, the lattice of all coarse-grainings grows super-exponentially with the number of values a process can take, and focusing only on causally closed processes provides a first important simplification. 
Furthermore, the mapping from the lattice of all causally-closed processes is found to be sometimes many-to-one, as different coarse-grainings may have an equivalent $\epsilon$-machine. 
This `degeneration' of the mapping gives rise to the notion of \emph{computational equivalence} (Definition~\ref{def:comp_equivalence}), which correspond to coarse-grainings that map into the same $\epsilon$-machines and hence are effectively carrying out the same computations, even if they look different.

Let us further explore the benefits of focusing on the hierarchy of strongly lumpable coarse-grainings of the $\epsilon$-machine instead of the hierarchy of causally closed coarse-grainings. Causal/information closure can take place in processes that are not causally-driven; for example, if $X_t=(Y_t,W_t)$ with $W_t$ being just noise (i.e.,\ an i.i.d.\ process), then $g(X_t)=W_t$ satisfies causal/informational closure while having no real causal structure. 
Our approach identifies such coarse-graining as being computationally trivial, being mapped to the top of the lattice of $\epsilon$-machines and hence being computationally equivalent to the constant coarse-graining (i.e. the mapping $g(\past{x}{t})=1$ for all trajectories $\past{x}{t}$, whose $\epsilon$-machine has a single causal state). 
More generally, the lattice of causally closed coarse-grainings fails to recognise levels that are computationally equivalent. For example, if $X_t=(U_t,V_t)$ with $U_t$ and $V_t$ being two copies of the same process (i.e.\ $U_t=\phi(V_t)$ with $\phi$ a bijection), then both $g_1(X_t)=U_t$ and $g_2(X_t)=V_t$ may be regarded as two distinct emergent levels --- while they should be considered to be copies of the same process.

The above examples illustrate how an analysis based solely on causal/informational closure can potentially overestimate the effective number of different emergent macro-scales that a system of interest may have.   
Indeed, it is the hierarchy of strongly lumpable coarse-grainings of the $\epsilon$-machines, and not the other ones, what determines what computations are taking place where by revealing which levels are running which software. 
For this reason, we propose \emph{to regard the hierarchy of strongly lumpable coarse-grainings of the $\epsilon$-machine provide as the natural blueprint of the functional architecture of a complex system}.

Let us illustrate the power of the hierarchy of $\epsilon$-machines as a guiding blueprint in a simple analysis. For this, let's consider a microscopic process $\bm X$, a causally closed macroscopic process $\bm Z$, and the trivial constant coarse-graining $\bm 1$ (which assigns the number 1 to all possible states of the micro process). These three processes form a totally ordered set in `real space' with $\bm X$ below, $\bm 1$ above, and $\bm Z$ between (indeed, $\bm 1$ is always a --- rather trivial --- coarse-graining of $\bm Z$ too). How does the resulting hiearchy of $\epsilon$-machines of these processes look like? Despite of the specific properties of $\bm X$ and $\bm Z$, an analysis based on the $\epsilon$-machines reveals that the computations done by $\bm Z$ can be described as falling into one of four possible scenarios (see Figure~\ref{fig:minimal_example}):
\begin{itemize}
    \item[a)] \textit{Non-trivial computation different from microscale}: the $\epsilon$-machine of $Z_t$ is non-trivial (i.e.\ different from the $\epsilon$-machine of $\bm 1$), and also different from the one of $X_t$.
    \item[b)] \textit{Non-trivial computation equivalent to microscale}: the $\epsilon$-machine of $Z_t$ is non-trivial, but is equal to the one of $X_t$.
    \item[c)] \textit{Trivial computation different from microscale}: $\epsilon$-machine of $Z_t$ is trivial and different from the one of $X_t$.
    \item[d)] \textit{Trivial computation equivalent to microscale}: $X_t$ is i.i.d., so there are no actual computations taking place in it or any of its coarse-grainings.
    \end{itemize}
Please note that while the lattice of coarse-grainings in `real space' is always the same, the lattice of the corresponding $\epsilon$-machines clearly illuminates the character of the computations being done.
\begin{figure}[t]
  \centering
  \if1\compiletikz
  \includetikz{tikz/}{MinimalExample}
  \else
    \includegraphics{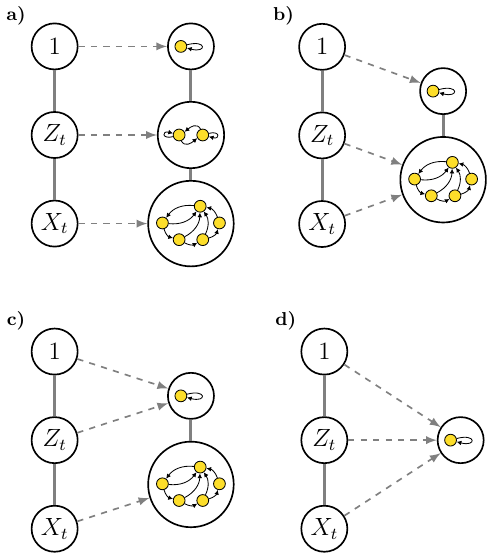}
  \fi 
  \caption{\textbf{Possible computational architectures of an emergent macroscopic level}. Our theory shows that the computations carried out by a causally closed process $\bm Z$ with respect to a microscopic process $\bm X$ and the trivial coarse-graining $\bm 1$ can be categorised within four groups, illustrated here. The computations are the same as the ones at the microscale if the $\epsilon$-machine of $\bm X$ and $\bm Z$ are equivalent (as in \textbf{b} and \textbf{d}), and are trivial if the $\epsilon$-machines of $\bm Z$ and $\bm 1$ are equivalent (as in \textbf{c} and \textbf{d}). At the left of each subplot is the lattice of coarse-grainings in real space, which is the same for the four cases; at the right is the lattice of corresponding $\epsilon$-machines in theory space, which better illustrates the effective computational structure of the system.}
\label{fig:minimal_example}
\end{figure}

In summary, our theory reveals that it is the hierarchy of strongly lumpable causal states, and not the one of causally/informationally closed spatial coarse-grainings, what provides the most illuminating blueprint of the software-like processes running over a give microscopic layer at which different computations take place. 
By grouping together all computationally equivalent processes, the lattice of the resulting $\epsilon$-machines is generally simpler than the other lattices, and provides a more accurate representation of the computations happening at different levels. These ideas are illustrated in a series of case studies presented in the next section.

\begin{figure*}
  \centering
  \if1\compiletikz
  \includetikz{tikz/}{eca}
  \else
  \includegraphics{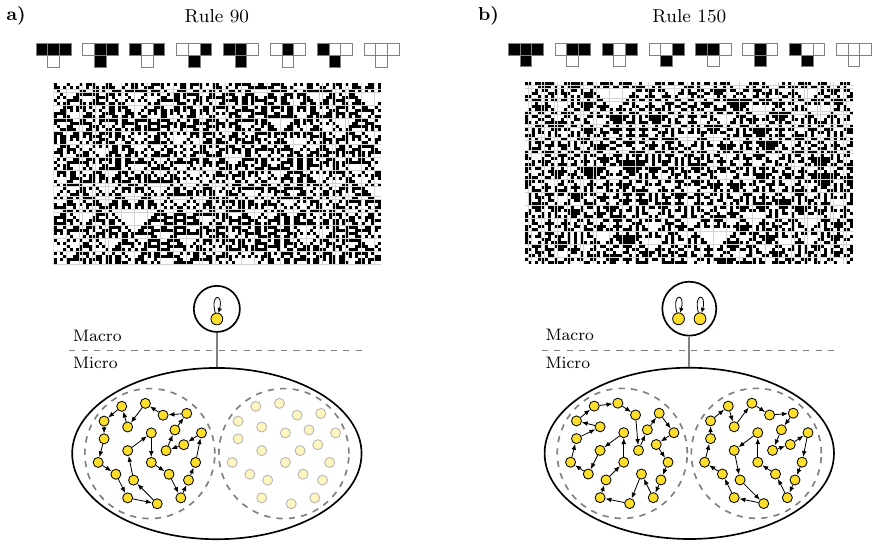}
  \fi 
  \caption{
  \textbf{Conserved quantities in elementary cellular automata}. Illustration of the computations associated to different types of conserved quantities. 
  \textbf{a)} Rule 60 forces configurations to have even parity. Hence, the parity is a conserved quantity which is computationally trivial, akin to case (c) in Figure~\ref{fig:minimal_example}. 
  \textbf{b)} In contrast, rule 150 keeps the parity of the initial condition. Hence, while the parity is also a conserved quantity for these dynamics, the computations associated with it are non-trivial, akin to case (a) in Figure~\ref{fig:minimal_example}. 
  }
    \label{fig:ECAs}
\end{figure*}

\section{Case studies}
\label{sec:examples}

This section presents a number of examples that help to illustrate, exemplify, and develop further intuitions about our theory. The investigated systems include various paradigmatic models of complexity including cellular automata, diffusion processes, Ising models, random walks over networks, and recurrent neural networks.

\subsection{Invariant quantities}

Let us start by considering a system whose microscopic evolution is
described by the time series $\bm X$, and has a scalar coarse-graining $Z_t = f(X_t)$ such
that $Z_t=Z_{t_0}$ for all $t$ --- i.e.,\ $\bm Z$ is constant for all values of $t$, being an invariant property of
$\bm X$. Such invariants may, or may not, depend on initial condition. A direct calculation shows that invariants are informationally closed, as
\begin{align*}
    I(\past{X}{t} ; \finfut{Z}{t+1}{L} | \past{Z}{t}) 
    &= H(\finfut{Z}{t+1}{L}|\past{Z}{t}) - H(\finfut{Z}{t+1}{L} | \past{X}{t} ) \\
    &= H(Z_{t_0} | Z_{t_0}) - H(Z_{t_0} | \past{X}{t})\\
    &=0~.
\end{align*}

Invariants may be computationally trivial (i.e.\ computationally equivalent to a constant coarse-graining, as in classes (c) or (d) in Fig.~\ref{fig:minimal_example}) or not (i.e.\ classes (a) or (b)), depending on their entropy. Specifically, a Markov process has invariants that are computationally
non-trivial when the state space of the Markov chain can be divided in
regions such that the process can move between states within individual
regions but not between states of different regions (i.e., if the Markov chain is reducible). If that is the
case, the identity of each isolated group is a non-trivial invariant,
and the corresponding coarse-graining looks like a Markov chain that only has 
self-transitions. In contrast, a computationally trivial
invariant is a way of describing the part of the phase space that
is not accessible. For example, if $Z_t=f(X_t)=c$ is a constant coarse-graining, that means that all micro-states $X'_t$ such that $f(X'_t)\neq c$ cannot occur.

Let us illustrate these ideas in an concrete example based on cellular automata (see Figure~\ref{fig:ECAs}). Consider an elementary cellular automaton whose state of its $n$ binary cells at time $t$ is denoted by $X_t = (X_t^{(1)}, \dots, X_t^{(n)})$~\cite{wuensche1992global}. Let's focus on the so-called `rule 150', whose temporal evolution is equivalent to a triple \texttt{xor} logic gate:
\begin{equation}
    X^{(k)}_{t+1} \equiv X_t^{(k-1)} + X_t^{(k)} + X_t^{(k+1)} \Mod{2}
\end{equation}
for $2\leq k \leq n-1$. For the sake of symmetry, let's also impose circular conditions, so that 
$X^{(1)}_{t+1} \equiv X_t^{(n)} + X_t^{(1)} + X_t^{(2)} \Mod{2}$ and
$X^{(n)}_{t+1} \equiv X_t^{(n-1)} + X_t^{(n)} + X_t^{(1)} \Mod{2}$. 
It can be shown (e.g.\ via direct numerical evaluation) that this system is stationary under the uniform distribution on the $2^n$ possible states. This implies that $I(X_t;X_{t-1}) = H(X_{t+1}) = n$, where the first equality is a consequence of the fact that the automaton's evolution is deterministic.

Consider now a coarse-graining that returns the parity of a given state, i.e.\ $Z_t \equiv \sum_{k=1}^n X_t^{(k)} \Mod{2}$. A direct calculation shows that this is a conserved quantity of the dynamics of rule 150, as
\begin{align}
    Z_{t+1} &\equiv \sum_{k=1}^n X_{t+1}^{(k)} \Mod{2} \\
    &\equiv \sum_{k=1}^n 3 X_t^{(k)} \Mod{2} \\
    &\equiv \sum_{k=1}^n X_t^{(k)} \Mod{2}\\
    &\equiv Z_t~.
\end{align}
Moreover, it can be seen that $H(Z_t)=1$, as the parity of the state
can be 0 or 1 with equal probability if $X$ is initialised at the
uniform distribution. Therefore, while the microscopic system computes
$n$ bits, $Z_t$ only computes one --- in this case, in the form of information storage~\cite{lizier2012local}. This shows that this invariant is not trivial and also not equal to the microstate, hence corresponding to class (a) in Fig.~\ref{fig:minimal_example}.

For a contrast, let us consider another elementary cellular automata: number 60, whose rule of evolution is equivalent to a single \texttt{xor}
\begin{equation}
    X_{t+1}^{(k)} \equiv X_t^{(k)} + X_t^{(k+1)} \Mod{2}.
\end{equation}
It can be shown that, for this case, the parity $Z_t \equiv \sum_{k=1}^n X_t^{(k)} \Mod{2}$ is always even, i.e.\ that $Z_t=0$. This makes $\bm Z$ an invariant whose associated computation, however, is trivial, representing class (c) in Fig.~\ref{fig:minimal_example}. This also reveals, as argued above, that certain microscopic states are not reached by the dynamics --- namely, the states with odd parity (see Figure~\ref{fig:ECAs}).

\subsection{Ehrenfest diffusion model}
\label{sec:ehrenfest}

Let us now study a computationally closed quantity that is not
invariant. For this, let us consider the Ehrenfest model~\cite[~\S{}7.3]{Kemeny_FMC}, which is a popular diffusion model in statistical mechanics. This model describes the behaviour of a gas in a container made of two interconnected chambers (see Fig.~\ref{fig:Ehrnfest}a). The state of the $n$ molecules of the gas at time $t$ is described by a $n$-dimensional binary vector $X_t = (X_t^{(1)},\dots,X_t^{(n)})$, where $X_t^{(k)} \in \{0,1\}$ determines in which of the two chambers the $k$-th molecule is located. The dynamics of the gas are then established as follows: 
\begin{itemize}
    \item At each timepoint, one of the $n$ molecules is chosen randomly with equal probability.
    \item With probability $q$, move the chosen particle to the other chamber. Otherwise, leave the molecule in its current chamber.
\end{itemize}
The system has two parameters: the total number of particles $n$, and
$q\in[0,1]$ that regulates how easy is for molecules to move between
the two chambers. Hence, if there are $z\in\{0,\dots,n\}$
particles in the first chamber, there is a probability of
$z/n$ of choosing one of them. Therefore, one can find that there are
three possible actions: (i) move a molecule from the first to the
second chamber with probability $qz/n$, (ii) move a particle from the
second to the first chamber with probability $q(n-z)/n$, or (iii) leave the particles
the way they are with probability $1-q$. These dynamics on
$\bm X$ are equivalent to a random walk (i.e., a Markov chain) on the edges of an
$n$-dimensional cube, as each step involves at modifying the state
along not more than one dimension.

\begin{figure}[!t]
  \centering
  \if1\compiletikz
  \includetikz{tikz/}{ehrenfest}
  \else
    \includegraphics{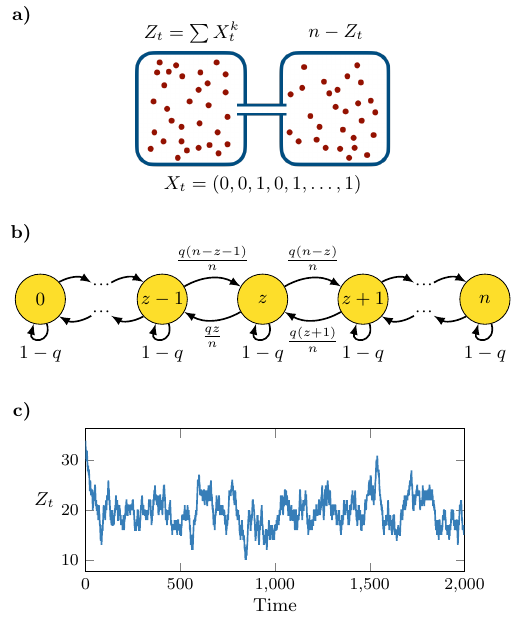}
  \fi 
  \caption{\textbf{Ehrenfest diffusion model.} \textbf{a)} The model considers particles contained in two connected chambers. The microscopic description of the system ($X_t$) is a binary vector that specifies in which container is each particle, while the macroscopic description ($Z_t$) is the number of particles in the left chamber.
  \textbf{b)} Illustration of the finite state machine description of the $\epsilon$-machine corresponding to the macroscopic variable. 
  \textbf{c)} One realisation of the dynamics of the macroscopic process of a system of $n=40$ particles, which naturally oscillates around $n/2$.}
  \label{fig:Ehrnfest}
\end{figure}

Let us now consider the following coarse-graining:
\begin{equation}
    Z_t = \sum_{k=1}^n X_t^{(k)}~,
\end{equation}
which tracks the number of molecules in each partition. 
At each timepoint, $Z_t$ can either increase by 1, decrease by 1, or
stay the same; however, for $n\gg 1$ there is a trend for $Z_t$ to move towards $n/2$, which makes the model a useful testbed to study thermalisation processes~\cite{bingham1991fluctuation}.

It can be shown that $\bm Z$ is a
(time-homogeneous) Markov chain~\cite{Kemeny_FMC}, which implies that only the number of particles --- and not their identities --- in each chamber matters for the evolution of $Z_t$. The finite state machine equivalent of $\bm Z$
is given in Figure~\ref{fig:Ehrnfest}b. Moreover, it can be shown~\cite[p.~170]{Kemeny_FMC} that $\bm Z$ is a Markov chain irrespective of the initial distribution of $\bm X$ --- i.e.\ that it is strongly lumpable (see Definition~\ref{def:lumpability}). Thanks to Proposition~\ref{prop:Markov_Lumpability}, this implies that $\bm Z$ is causally closed in the sense of Definition~\ref{def:causal_closure}, and furthermore is also computationally closed thanks to Theorem~\ref{teo:compu_closure}.

The implications of these results can be contemplated from various angles. 
From a interventional perspective, this implies that if one wants to
control the evolution of the number of particles in each chamber, one does not need to worry about the exact
placement of individual particles $X_t$ beyond what is specified by
$Z_t$ --- that is, setting $Z_t$ gives all the control that is possible about future values. 
From the point of view of simulations, this implies that one can
investigate the dynamics of $\bm Z$ (as shown in Fig.~\ref{fig:Ehrnfest}b) by direct simulation of their Markov
dynamics while disregarding $\bm X$. 
This is guaranteed by the fact that, thanks to causal closure, including the full state $\bm X$ into the simulation adds nothing relevant to the evolution of $\bm Z$. 
Finally, from an theoretical perspective, these results imply that this system has a non-trivial macroscopic level that runs and processes information irrespective of the microscopic details in a software-like manner. Following Proposition~\ref{pro:info_decomposition}, the total information processed by the system at each time-step (given by $I(\bm X_t; \bm X_{t+1})$) can be decomposed into the portion that is processed at the macro level (given by $I(\bm Z_t; \bm Z_{t+1})$) and the one that is not (given by $I(\bm X_t; \bm X_{t+1})-I(\bm Z_t; \bm Z_{t+1})$). 
Overall, these findings enable a deeper view on the dynamics of the Ehrenfest model, revealing specific ways in which it can be efficiently predicted and controlled.

\subsection{Ising model with Glauber dynamics}
\label{sec:ising}

Let us now study another open system, this time including interactions between its microscopic components. Consider a collection of $n$ magnetic spins, whose state at time $t$ is described by the vector
$X_t = (X_t^{(1)},\dots,X_t^{(n)})$, where $X_t^{(k)} = 1$ if the
$k$-th spin is pointing up and $X_t^{(k)} = -1$ if it is pointing
down. The energy of the system on state $X_t$ is determined by a
Hamiltonian $\mathcal{E}(X_t)$. Furthermore, let's consider standard
Glauber dynamics~\cite{glauber1963time,levin2017markov} (which is
closely related with the so-called kinetic Ising
model~\cite{aguilera2021unifying})  observing the following transition
probability:
\begin{equation}\label{eq:prob_glauber}
    p_{X_{t+1}|X_{t}}( x' | x)=
\begin{cases}
\frac{1+ \tanh{ \beta \Delta\mathcal{E}}}{2n} & \text{if } x' \in \mathcal{B}(x),\\
1 - \sum_{x'\in\mathcal{B}(x)} \frac{1+\tanh{\beta\Delta{E}}}{2n} &\text{if }x'=x,\\
0 &\text{otherwise,}
\end{cases}
\nonumber
\end{equation}
where $\mathcal{B}(x)$ are the set of configurations that differ from $x$ in exactly one spin, and $\Delta\mathcal{E} = \mathcal{E}(x)-\mathcal{E}(x')$~\footnote{Please note that one could sub-sample the temporal evolution of this process (so to account for updates of $k$ spins instead of only one), and the arguments that follow still hold.}. These statistics can be conceived as arising from the following procedure:
\begin{itemize}
\item While being at state $x$ at timepoint $t$, choose one of the $n$ spins
  randomly with equal probability. Denote by $x'\in\mathcal{B}(x)$ the
  state that is equal to $x$ except with one selected spin being
  flipped.
\item Calculate the energy difference between $x$ and $x'$ 
  as $\Delta\mathcal{E} = \mathcal{E}(x)-\mathcal{E}(x')$.
\item Adopt $x'$ as state in timepoint $t+1$ with probability
  $p = (1+ \tanh{ \beta \Delta\mathcal{E}})/2$, otherwise stay in $x$.
\end{itemize}
It can be shown~\cite{glauber1963time} that those dynamics satisfies detailed balance, and gives rise to a unique stationary distribution that corresponds to the standard Boltzmann-Gibbs distribution given by 
\begin{equation}
    \tilde{p}(x) = \frac{e^{-\beta \mathcal{E}(x)}}{\mathcal{Z}}~,
\end{equation}
with $\mathcal{Z}=\sum_x e^{-\beta \mathcal{E}(x)}$ being the partition function.

Let us now consider the coarse-graining $Z_t = \mathcal{E}(X_t)$, which establishes the dynamics of the energy of the system. 
For simplicity of the analysis, let's assume that the system is such that its Hamiltonian is invariant under spin permutations~\footnote{This is the case, for example, when the systems observes a uniform and fully-connected graph of interactions $J_{i,k}$.}. 
Using this symmetry, it is possible to show that
\begin{equation}
p_{Z_{t+1}|X_t}(z|x) = \psi\big(z,\mathcal{E}(x)\big)~,
\end{equation}
where $\psi(z,e)$ is a function of the binary vector $z$ and the scalar $e$. This implies that the $\upsilon$-machine of $\bm Z$ is equal to its $\epsilon$-machine --- i.e.\ that $\bm Z$ is causally closed. 
Thanks to Theorem~\ref{teo:compu_closure}, this in turn implies that $\bm Z$ is also computationally closed. It can be shown $\bm Z$ is generally not computationally equivalent to $\bm X$ as the $\epsilon$-machine of $\bm X$ is generally richer than the one of $\bm Z$ (as the specific configuration $x$ often makes a difference for $p_{X_{t+1}|X_t}(x'|x)$), hence belonging to class (a) in Figure~\ref{fig:minimal_example}.

\begin{figure}[t!]
  \centering
  \if1\compiletikz
  \includetikz{tikz/}{ising}
  \else
    \includegraphics{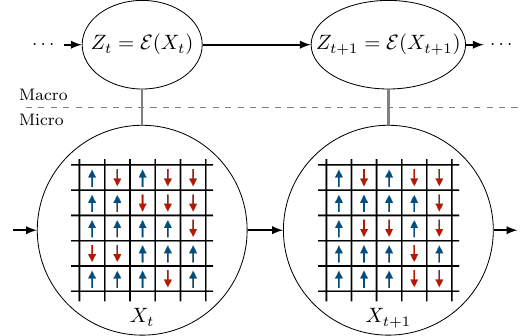}
  \fi 
  \caption{\textbf{Energy dynamics of an Ising model are causally closed.} When considering the Ising model under Glauber dynamics, it can be shown that its energy is a macroscopic variable whose dynamics are causally --- and hence also computationally --- closed.
  }
  \label{fig:Ising}
\end{figure}

The fact that the Hamiltonian of an Ising model is computationally
closed under Glauber dynamics implies that the dynamics of the energy
does not depend on microscopic details, but only on the current energy
level. This has similar implications for simulations and interventions
as in the case of the Ehrenfest model: a simulation of energy dynamics
in absence of its microscopic instantiation can be exact, and macroscopic
interventions to drive its future evolution can exhert as much control as it is possible.

\subsection{Random walks over modular networks}

Let us now study a system that has several nested levels of causally
closed dynamics. For this, consider a random walk over a weighted
network with a modular structure, made by $M$ strongly
connected communities of nodes that are weakly connected between each other. For
simplicity, let's assume that the links of this network are
characterised by two parameters: the weight of the connections within
individual modules, $r_\text{w}$, and the weight of the connection
between different modules, $r_\text{b}$, with $r_\text{w} \gg r_\text{b}$.
Additionally, let's consider a random walker which moves from node to node on this
network over discrete time-steps; it decides to move to another node with 
probability $q$, and chooses its destination with probability 
proportional to the strength of the link connecting them.

Let us calculate the statistics of such random walk. For this, 
let's denote as $c_x$ the label of the community to which node $x$ belongs and $\mathcal{N}$ the set of nodes of the network. 
Consider $X_t\in \mathcal{N}$ to be the state of the random walker at time
$t$, whose transition probabilities are
\begin{equation}
    p_{X_{t+1}|X_{t}}( x' | x) = 
    \begin{cases}
        1-q \quad &\text{if }x = x', \\
        q r_\text{w} /\phi(c_x) \quad &\text{if $c_x= c_{x'}$,} \\
        q r_\text{b} /\phi(c_x) \quad &\text{if $c_x \neq c_{x'}$.}
    \end{cases}
\end{equation}
Above, $\phi(c_x)= r_\text{w} |c_x| + r_\text{b} (n-|c_x|-1)$ is a normalisation term to guanrantee that the probabilities sum up to one, which depends only on the size of the current community, denoted by $|c_x|$. With this conditional probability, the resulting random walk will tend to roam around a given module, and will jump between modules only occasionally due to the significantly greater connectivity within
than between modules.

Let us now consider the following coarse-graining: $Z_t = c_{X_t} = f(X_t)$, where $f:\mathcal{N}\to \{1,\dots M\}$ is a function that returns the index of the module to which the given node --- its argument --- belongs. One can show that 
\begin{equation}
    p_{Z_{t+1}|X_{t}}( z | x) = 
    \begin{cases}
         |z| \cdot r_\text{w} /K \quad &\text{if $z= f(x)$,} \\
         |z| \cdot r_\text{b} /K \quad &\text{if $z \neq f(x)$.} \\
    \end{cases}
\end{equation}
where $z\in\{1,\dots,M\}$ is a community index, $|z|$ its size, and $K$ is a normalising
constant. The fact that $p_{Z_{t+1}|X_{t}}( z | x)$ only depends on
$f(x)$ but not on $x$ itself shows that $\bm Z$ is causally closed. 
Furthermore, thanks to Theorem~\ref{teo:compu_closure}, $\bm Z$ is also computationally closed. Additionally, if the probability of staying in the same node ($1-q$) is different from the one moving to another node in the same community, then $X$ and $Z$ are not computationally equivalent. Indeed, in that case $p_{X_{t+1}|X_t}(x'|x)$ changes for different nodes $x$ within the same community, making the $\epsilon$-machine of $\bm X$ have more states than the $\epsilon$-machine of $\bm Z$ --- and hence belonging in class (a) in Figure~\ref{fig:minimal_example}.

Under some circumstances, this system exhibits multiple nested levels that are causally closed. To see this, let's consider the coarse-graining $W_t = h(Z_t)$ which identifies communities of different cardinalities but does not disambiguate between communities of the same size. For example, if the network considers two communities of 10 elements and three of 20 elements, $W_t$ only discriminates if the random walker is located in a community with 10 elements or in one with 20. Following an argument analogous to the one made with respect to $\bm Z$, it can be shown that $\bm W$ is also causally closed: by considering the conditional probability $p_{W_{t+1}|X_t}(w|x)$, one can show that it depends only in $h\big(f(x)\big)$, which shows that the $\upsilon$-machine of $\bm W$ is equal to its $\epsilon$-machine and hence it is causally close. Hence, thanks to Theorem~\ref{teo:compu_closure}, $\bm W$ is also computationally closed.

\begin{figure}[t!]
  \centering
  \if1\compiletikz
  \includetikz{tikz/}{network}
  \else
    \includegraphics{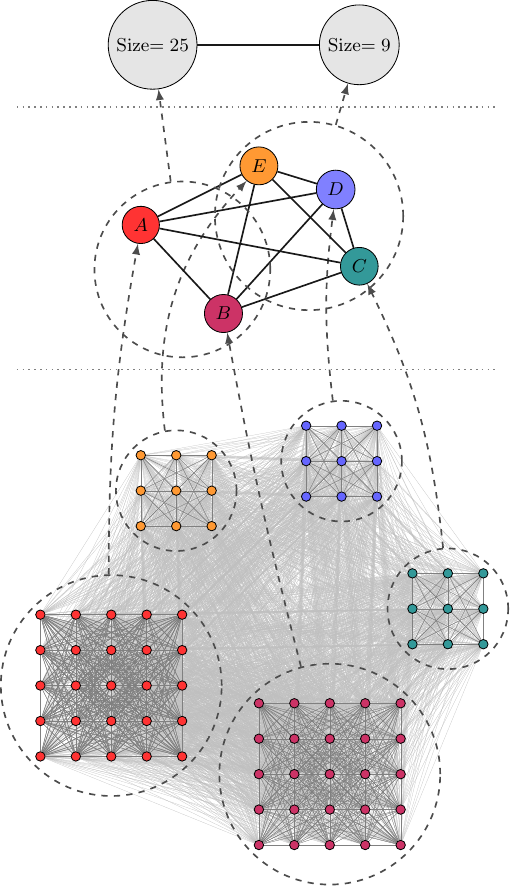}
  \fi 
  \caption{\textbf{Causally closed coarse-grainings of a random walk over a network.} 
  A random walk on a modular network can be coarse-grained such that the dynamics over the module's labels is causally closed. Furthermore, by considering equivalence classes of modules given by their size provides a further causally closed macroscopic process.}
  \label{fig:network}
\end{figure}

This analysis shows that the computations involved in this random walk can be construed as taking place in (at least) three nested levels: first the cardinality of the future community is determined, then the specific identity of the community is chosen, and finally the specific node within that community is decided. Crucially, Proposition~\ref{pro:info_decomposition} states that the information processed by the computations taking place at high levels is done in isolation of what happens in the lower ones. This has similar consequences as for previous examples: macroscopic dynamics can be predicted and simulated in isolation of microscopic details, and macroscopic interventions guarantee optimal macroscopic controllability.

\subsection{Agent-based simulation}
\label{sec:agent_based_modelling}

Further examples for systems with computational closure can be found
in the literature of agent-based models, which provides useful tools to investigate collective human and animal behaviour~\cite{crooks2011introduction}.
Agent-based models are usually Markov chains defined over large state-spaces, which establish the possible configurations for a collection of interacting agents. 
In most cases the aim of a simulation is not to identify the exact behaviour of individual agents, but rather to characterise their aggregate statistics --- e.g., the number of agents in a given state. These aggregate statistics can be seen as 
coarse-grainings $\bm Z$ of the underlying agent-based model $\bm X$.

The literature has studied settings under which agent-based models are
lumpable --- i.e., conditions under which the sequence of aggregate
statistics is itself a Markov chain. Examples for this line of work
are studies on opinion dynamics such as the voter
models~\cite{Banisch_ABM,Lamarche-Perrin_ABM} or agent movement on
graphs~\cite{Geiger_MarkovFlee}. Lumpability is a useful feature in
agent-based modelling, as it implies that one can avoid running
the whole model $\bm X$ and instead run a reduced Markov chain that gives the evolution on the coarse-graining $\bm Z$. Indeed, as shown in Proposition~\ref{prop:Markov_Lumpability}, in this case the coarse-graining $\bm Z$ is computationally closed.

To illustrate this principle, let us follow Ref.~\cite{Geiger_MarkovFlee} and consider an agent-based model that can be used to investigate agent mobility, e.g.\ to simulate the movement of refugees in conflict locations. 
In this model, $n$ agents randomly move on a graph with $L$ nodes, and $X_t=(X_{t}^1,\dots,X_{t}^n)$ denotes the microscopic state of the system where $X^i_{t}\in\{1,\dots,L\}$ denotes the location of the $i$-th agent at timepoint $t$. At each iteration, an agent decides to leave their current vertex with some probability depending on the attributes of the current vertex. If the agent decides to leave, then it  moves to a neighbouring vertex with a probability depending on the attributes of the respective location and the weight of the connecting edge. This procedure is repeated for all agents and for a predefined number of iterations. 
Simulating the full model has a computational cost that scales linearly with the number of agents. However, the main aim of refugee movement simulation is to identify not individual but collective behaviours --- for example to determine the dynamics of the number
of agents at particular locations (e.g., humanitarian camps). Hence,
the goal is not to study the dynamics of so-called `world configuration' 
$X_t$, but rather to infer properties of a coarse-graining
$Z_t$ defined as the vector of the number of agents at
specific locations. Mathematically, $Z_t=(Z^1_{t},\dots,Z^L_{t})$, where  $Z^\ell_{t}=\sum_{i=1}^n | \{i:X^i_{t}=\ell \}|$ with $|A|$ denoting the cardinality of the set $A$.

Importantly, it has been shown that --- under certain conditions ---
the original agent-based model $\bm X$ is lumpable w.r.t.\ this
coarse-graining, i.e., that the conditional probability of the agent population vector state $Z_t$ depends only on the agent population $Z_{t-1}$ of the previous iteration, and not on the world state $X_{t-1}$~\cite{Geiger_MarkovFlee}. Thanks for Theorem~\ref{teo:lumpability}, it follows that the coarse-graining $\bm Z$ is informationally
closed, and due to Theorem~\ref{teo:compu_closure} it is also computationally closed. 

As a consequence, rather than simulating the full
agent-based model, it is sufficient to simulate the re-distribution of
populations over vertices, i.e., to simulate $\bm Z$. Importantly, the
computational complexity of simulating $\bm Z$ does not depend
strongly on the number of agents $n$, but rather on the number of
locations $L$, which is typically much smaller than the number of
agents. This example illustrates an important aspect of computational closure: computationally closed coarse-grainings may allow exact simulations --- and hence forecasting and predictions --- that may be unfeasible to run otherwise.

\subsection{Memory retrieval and attractor dynamics}
\label{sec:hopfield}

To conclude, let us study a recurrent neural network architecture that models how associative memory works in the human brain~\footnote{For a comprehensive review on the state of the art of the neuroscience of memory, see Ref.~\cite{kandel2014molecular}.}. 
For this, we will focus on the well-known Hopfield network model~\cite{amari1972learning,little1974existence,hopfield1982neural}, which provides a simple but foundational instantiation to the long-standing principle that memories are not stored inside individual neurons, but rather in their synapses~\cite{ramon1894fine,hebb1949organisation}.

Following Ref.~\cite{gerstner2014neuronal}, let us consider a Hopfield network of $n$ neurons, which is a system with Markovian dynamics where neurons are modeled as binary units that can be in states \emph{on} ($X_t^j=1$) or \emph{off} ($X_t^j=-1$). Neurons interact with each other via synapses, whereas the weight of the synapsis from the $i$-th towards the $j$-th neuron is denoted by $w_{i,j}$. The input potential of the $j$-th neuron at time $t$ is then given by
\begin{equation}
    h_t^j = \sum_i w_{i,j} X_t^j~.
\end{equation}
This potential determines the update of the $j$-th neuron as follows:
\begin{equation}\label{eq:prob_loc}
    \mathbb{P}\big\{X_{t+1}^j=1|X_t\big\}
    = g\big(h_t^j\big)~,
\end{equation}
where $X_{t} = (X_t^1,\dots,X_t^n)$ and 
$g$ is a monotonous non-linear function of $h_t^j$. A
common choice is 
\begin{equation}
g(h) = (1+\tanh{\beta h})/2,
\end{equation} 
with the inverse temperature parameter $\beta$ controlling the stochasticity of the
system. In particular, $\beta\to\infty$ leads to deterministic dynamics
given by $g(h)=\big(1+\text{sign}(h)\big)/2$. 
Furthermore, for simplicity it is assumed that the dynamics of the whole system are simply a combination of the individual dynamics, i.e.
\begin{equation}\label{eq:prob_glob}
     \mathbb{P}\big\{X_{t+1} = x_{t+1} |X_t\big\} = \prod_j  \mathbb{P}\big\{X_{t+1}^j=x_{t+1}^j|X_t\big\}.
\end{equation}

How can a Hopfield network store and retrieve memory patterns?
Let us denote the patterns to be stored as $p^1,\dots,p^m$, with $p^\mu=(p_1^\mu,\dots,p_n^\mu)$. Then, one defines the synaptic weights as follows:
\begin{equation}
    w_{i,j} = \frac{1}{n}\sum_{\mu=1}^m p_i^\mu p_j^\mu~.
\end{equation}
Using those weights, the resulting dynamics of the Hopfield network can be shown to have the
mentioned patterns as attractors. This can be seen by noting that the
energy function 
\begin{equation}
    \mathcal{E}(X_t) = -\sum_i\sum_j w_{i,j} X_t^i X_t^j
\end{equation}
is a Lyapunov function of the deterministic dynamics
($\beta\to\infty$), and the $m$ patterns $p^j$ constitute its local minima (see Ref.~\cite{gerstner2014neuronal}). 
Please note that the number of patterns a network can retain depends on
the stochasticity of the system and the desired level of reliability,
and is usually much smaller than the number of neurons~\cite{amit1985storing,amit1987statistical}.

Let us now consider a macroscopic description of the system that is based on the memory patterns instead of neurons. Following Ref.~\cite[Sec.~17.2.3]{gerstner2014neuronal}, we consider the projection of the current neural activity on the $j$-th memory pattern as given by the dot product between them:
\begin{equation}
    Z_t^\mu = \frac{1}{n} \sum_i p_i^\mu X_t^i~.
\end{equation}
Note that $-1\leq Z_t^\mu \leq 1$, with 
\begin{itemize}
    \item $Z_t^\mu=1$ if $X_t$ and $p^\mu$ are equal, 
    \item $Z_t^\mu=-1$ if $X_t$ and $p^\mu$ are opposite, and
    \item $Z_t^\mu\approx 0$ if $X_t$ and $p^\mu$ are
      uncorrelated (for large $n$, which we assume here).
\end{itemize}
Finally, we define $Z_t = (Z_t^1,\dots,Z_t^m)$ as the vector of all $m$ projections at time $t$ (see Figure~\ref{fig:hopfield}).

\begin{figure}[t!]
  \centering
  \if1\compiletikz
  \includetikz{tikz/}{hopfield}
  \else
    \includegraphics[width=\columnwidth]{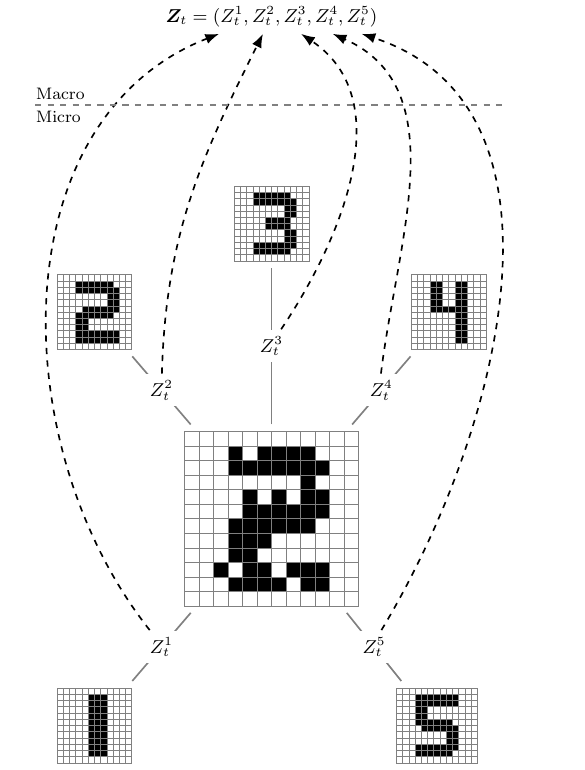}
  \fi 
  \caption{\textbf{Hopfield network compute memory retrieval on a causally closed macroscopic level}. The state of a Hopfield network is determined by the activity of each of the involved neurons, here represented as a square grid. Nonetheless, the similarity between the present pattern and the patterns that the network stores (denoted by $Z_t^\mu$, with $\mu \in \{1,2,3,4,5\}$ in the figure), which determines to which of the stored patterns is more similar to the current configuration. Our results show that $Z_t = (Z_t^1,\dots,Z_t^5)$ is a causally closed coarse-graining of the neural system, which critically determines the memory retrieval process.}
  \label{fig:hopfield}
\end{figure}

Let us show that $\bm Z$ is computationally closed. 
For this, first we show that the neuron potential can be
re-written in terms of these projections as follows:
\begin{equation}
     h_t^j = \sum_i w_{i,j} X_t^j = \frac{1}{n}\sum_i\sum_\mu p_j^\mu p_i^\mu X_t^i = \sum_\mu p_j^\mu Z_t^\mu.
 \end{equation}
This, combined with Eqs.~\eqref{eq:prob_loc} and \eqref{eq:prob_glob}, imply that the transition probabilities $\mathbb{P}\big\{X_{t+1} = x_{t+1} | X_t\big\}$ depend only on $Z_t$, which in turns implies that $H(Z_{t+1}|X_t) = H(Z_{t+1}|Z_t)$ guaranteeing that $\bm Z$ is informationally closed. This, in turn, implies --- due to Theorem~\ref{teo:compu_closure} --- that $\bm Z$ is computationally closed with respect to $\bm X$.

These results imply that Hopfield networks implement processes of memory retrieval --- more specifically, the calculation of the similarity between the present and the stored patterns --- at a macroscopic level, such that its processing of information takes place independently of the microscopic level of individual neurons.

\section{Discussion}
\label{sec:discussion}

Understanding the nature of emergent macroscopic phenomena is a fundamental theoretical challenge, and also central for advancing in important open questions in various branches of science. 
In this paper we introduce a theory that investigates emergent phenomena from a computational perspective, illuminating their inner workings via core principles of theoretical computer science. 
Specifically, the theory introduced in this paper explains what it means for a natural system to display software-like macroscopic processes in a way that is theoretically sound and practically useful. In the sequel, we highlight various implications of the theory, while also discussing related work and outlining future research directions.

\subsection{A computational blueprint of the functional architecture of complex systems}

By taking inspiration on how software works, our approach articulates a view on emergence in terms of macroscopic processes that are self-contained in causal, informational, and computational terms. 
Our mathematical formalism revealed a single condition (namely, the equivalence between $\epsilon$-machines and $\upsilon$-machines) under which macroscopic processes resulting from spatial coarse-grainings guaratee all three properties. 
Processes satisfying these properties possess a number of interesting features. 
First, due to causal closure, all interventions that drive such processes are accounted by events of the processes themselves --- and, in that sense, it can be said that such a process `causes itself.' 
Second, due to information closure, microscopic events don't have any additional predictive power over such process that couldn't be found in the process itself. This implies that optimal prediction of such macroscopic processes can be attained while disregarding all details about their microscopic instantiation. 
Finally, due to computational closure, such processes run a particular subset of the computations carried out by the whole system in an autonomous manner. 
As a consequence of this, such processes can be thoroughly described without relying on the full microscopic theory, but using solely a simpler, coarse-grained, self-contained version of it.

In addition of characterising the properties of emergent macroscopic processes individually, our theory also accounts for the properties of them as a group. Indeed, our formalism revealed that the collection of all emergent processes are hierarchically organised into a lattice. Furthermore, our results show that this lattice can become substantially simpler when considered not in `real space' (i.e.\ in terms of coarse-grainings of states) but in `theory space' (i.e.\ as nested computational machines), which only accounts for processes that are computationally distinct. 
Our results show how this hierarchy of nested computational machines delineates a blueprint for the functional architecture of the system, while also characterising the computations running at each level.

Overall, the theoretical framework introduced in this paper establishes a rigorous foundation to study emergent processes akin of human-designed software on natural systems. 
Specifically, our theory provides concrete criteria to answer under what conditions natural systems can be said to be `running software', and what are the consequences of this.

\subsection{The benefits of emergent macroscopic processes}

The examples investigated in Section~\ref{sec:examples} provide some hints of why complex natural systems may exhibit macroscopic levels that are software-like.

First and foremost, causally closed levels can be efficiently controlled from just macroscopic interventions, without needing to intervene on microscopic conditions. For example, if one needs to control the dynamics of the number of particles in each chamber in the Ehrenfest model (see Section~\ref{sec:ehrenfest}), or the total energy in an Ising model (see Section~\ref{sec:ising}), one can do this to the highest possible precision by setting an initial condition that only constrains the macroscopic property in question but otherwise leaves the remaining details unspecified/random. 
These results have an important practical implication: these macroscopic thermodynamic quantities can be optimally controlled by the actions of macroscopic organisms, whose scope of actions is restricted to set values on macroscopic variables. 
This property may explain the fact that macroscopic organisms can be highly causally-effective on their environments without the need of having to intervene on microscopic variables.

Additionally, the analysis done over Hopfield networks (Section~\ref{sec:hopfield}) opens the door to interesting analyses of biological process in general, and neural systems in particular. Our results show how the neural dynamics associated to a process of memory retrieval can take place at a macro level, being driven by macroscopic neural patterns while being independent of the implementation of these on the firing of individual neurons. As such, this analysis establishes an objective procedure to evaluate when a neural process is implementing a computation at a higher level, having interesting parallels with the classic distinction of levels of analysis introduced by Marr~\cite{marr2010vision}. 
Also, these results provide a concrete operationalisation to the notion of multiple realisability~\cite{bechtel1999multiple,polger2016multiple}, articulating how equivalent software-like process may take place over different microscopic systems (see Section~\ref{sec:softwareness}). However, this approach does not satisfies --- at least in principle --- the idea of substrate independence~\cite{bostrom2003we}, as the causal relationships of the microscopic process determine which macroscopic processes can take place over it, and hence some processes may only be compatible with specific types of substrates. 
Building on these insights, future work may investigate the applicability of such approaches to identify algorithms implemented on different biological processes --- neural and beyond.

Finally, causally closed processes have interesting properties not only for the system itself, but also for the scientist that study them. As remarked in the study of agent-based modelling (Section~\ref{sec:agent_based_modelling}), computationally closed levels allow the simulation of aspects of these processes in an efficient yet precise manner. 
Indeed, the evolution of causally closed levels can be thoroughly simulated without accounting for the microscopic details that underlie it, which can enable efficacious exact simulation procedures of processes that may be unfeasible to run otherwise.

\subsection{Related literature and future work}

The theory introduced here extends the framework of \emph{dynamical independence} proposed in Ref.~\cite{barnett2021dynamical}, from which it takes direct inspiration. 
Dynamical independence proposes to identify emergent macroscopic process via information closure, inheriting its strengths and weaknesses (for example, as discussed in Section~\ref{sec:lattice_of_emergence}, 
noise processes may satisfy dynamical independence while having no underlying causal structure). 
One way to circumvent these limitations is by requiring information closure to be `non-trivial', i.e.\ for emergent levels to actually carry some information processing~\cite{bertschinger2006information}. 
In this work we take a different approach by interpreting this limitations as symptoms of a deeper issue: that information-theoretic metrics alone are limited in the degree they can assess the mechanisms driving the data-generating process. 
This greatly limits the capability of these frameworks not only to determine if a macro-variable is emergent or not, but to provide a more detailed assessment of how this may be taking place. 
Importantly, it has been shown that statistical patterns can be qualitatively different from their data-generating mechanisms~\cite{rosas2022disentangling}. 
By embracing principles of computational mechanics and the related notion of \emph{intrinsic computation}~\cite{feldman2008organization,crutchfield2010introduction}, our proposed theory takes a first step in the direction of characterising the mechanisms driving emergent processes from a computational perspective.

The framework presented here have similar motivations to recent work reported in Ref.~\cite[Sec.~VIII]{rupe2023principles}. This approach characterises the degree of dependency between microscopic and macroscopic levels combining principles of computational mechanics with insights and formal tools from fluid dynamics and dynamical systems theory, being highly complementary with the approach and tools developed here and well-suited for analysing systems with continuous variables. Our framework also has similar motivations and aims (although relying on a radically different approach) with work reported in Ref.~\cite{costa2023maximally}, which introduces scalable methods to model hidden dynamics driving time series data and hence being promising in terms of its practical applicability. Future work may try to combine the complementary strengths of these approaches.

More broadly, there is an extensive literature considering efficient representations of microscopic processes, which selects macroscopic processes via a variety of criteria. Classic approaches include `lumpability' techniques (see Section~\ref{sec:lumpability} and Refs.~\cite{auger2008aggregation,simon1961aggregation,white2000lumpable}) and other aggregation approaches (e.g.\ Ref.~\cite{munsky2006finite,garcia2010minimal}). 
An interesting angle on this problem is taken by Ref.~\cite{wolpert2014optimal}, which considers coarse-graining processes that balance predictability at the macro level with the degree of computational complexity involved in the coarse-graining mapping relating microscopic and macroscopic levels. 
This idea is related to a definition of emergence in terms of computational mechanics properties presented in Ref.~\cite[Sec.~11.2.1]{shalizi2001causal}, based on the ratio between the information processed by coarse-grainings and the informational cost of their $\epsilon$-machine (i.e.\ the entropy of the causal states). 
Including the computational complexity of the coarse-graining, and potentially linking it with thermodynamic considerations, are important avenues for future work.

The framework presented here provides a natural complement to the work presented in Ref.~\cite{horsman2014does}, which describes natural computation as the dual of prediction --- i.e., while prediction is a logical procedure that tell us the result of the evolution of a system without running it, computation corresponds to the act of letting a system run so that its evolution provides an answer a query.
While our theory can also be interpreted in terms of relationships between `physical' and `abstract' processes and their predictability, our framework is oriented to highlight the role of micro-to-macro relationships, being closer in spirit to the recent literature on information-theoretic methods to quantify emergence~\cite{seth2010measuring,hoel2013quantifying,klein2020emergence,rosas2020reconciling,varley2021emergence,barnett2021dynamical} and other theoretical work focused on the relevance of coarse-grainings~\cite{flack2017coarse}. 
Thus, further developing the notion of `computation' implied in the present framework is a relevant avenue of future work, either by working along the lines of Ref.~\cite{horsman2014does}, or by further developing the framework with richer constructs from the theoretical computer science literature --- for instance, by considering other machines from Choamsky's computational hierarchy~\cite{hopcroft2001introduction} instead of focusing solely in deterministic automata.

It is important to highlight that while the present work follows a characterisation of emergence similar to the one investigated by Ref.~\cite{barnett2021dynamical}, other characterisations of emergence exist (see e.g. Refs.~\cite{seth2010measuring,hoel2013quantifying,rosas2020reconciling}). An interesting direction for future work is to investigate the feasibility of operationalising other views on emergence in terms of computational mechanics. 
Also, in this work we conceive causality following the computational mechanics literature, following the principle of `differences that make a difference.' 
However, the quantitative characterisation of causation is a challenging subject, with a wide literature and no agreed-upon approach for its quantification (see e.g.\ \cite{janzing2013quantifying,kocaoglu2017entropic,albantakis2019caused,comolatti2022causal}). Note that while some of the proposed metrics establish causation at the level of individual events, the framework introduced here identified causation of variables in terms of other variables. Extensions from variables to individual events could be attempted via point-wise information-theoretic metrics (see e.g.~\cite{finn2018pointwise,prokopenko2014transfer}), being an interesting avenue for future work.

Some of our results related to causal and informational closure are related with work reported in Ref.~\cite{Pfante_LevelID}, where the authors investigated how informational closure relates to other properties, including causal closure (which they call `commutativity') and Markovianity. Their results imply that 
informational closure is equivalent to causal closure for spatial coarse-grainings, a result we generalise to non-Markovian processes and spatio-temporal coarse graining functions in Theorem~\ref{thm:equivalence}~\footnote{Please note that the definition of informational closure used in Ref.~\cite{Pfante_LevelID} only requires $I(Z_{t+1};X_t|Z_t)=0$ due to Markovianity of $\bm X$, and is hence identical to the characterisation of strong lumpability in~\cite[Th.~9]{GeigerTemmel_kLump}.}.

Finally, it is important to remark that it is not straightforward to apply the present theory to empirical data of large systems. The main limitation is the practical estimation of potentially large $\epsilon$-machines. There is promising work in this direction~\cite{tino2001predicting,shalizi2014blind,marzen2014circumventing}), particularly exploiting the relationship between $\epsilon$- machines and the information bottleneck~\cite{tishby2000information}, and recent extensions of state-space modelling~\cite{gu2023mamba}. 
We leave it to future work to develop suitably efficient and robust estimations procedures.

\vspace{-0.3cm}
\subsection{Final remarks}
\vspace{-0.1cm}
Given the central position that the notion of `computation' has in our modern scientific worldview, it is natural to wonder what a computational perspective can offer to our fundamental understanding of natural processes studied by physics, chemistry, and biology. 
While computer science provides a rigorous and fundamental theory based on the construct of Turing machines~\cite{turing1936computable},  the application of these ideas to concrete scientific scenarios (e.g.\ in physics or neuroscience) is not always straightforward~\footnote{An illustrative example is the old but still ongoing debate in neuroscience of whether the brain ought to be conceived as a computer~\cite{fodor1975language,pylyshyn1986computation,rescorla2017ockham,milkowski2013explaining,richards2022brain} or not~\cite{van1995might,van1998dynamical,smith2005cognition}: while under some definitions the brain is clearly a computer and under others it is not, it is not entirely evident how any of these outcomes advances our actual understanding of how the brain works.}. 
We hope that the theory introduced in this paper may serve as a first step towards a thorough characterisation of emergent processes in a way that advances theoretical discussions and also facilitate breakthroughs in empirical scientific questions.

\section*{Acknowledgements}

We thank Patricio Orio for helpful discussions about cellular automata, and Lionel Barnett, Sebastian Dohnany, and Clemens von Stengel for useful comments about the manuscript. 
F.R. was supported by the Fellowship Programme of the
Institute of Cultural and Creative Industries of the University of Kent, and the DIEP visitors programme at the
University of Amsterdam. The work of B.G. was funded by the European Union’s Horizon Europe research and innovation programme within the KNOWSKITE-X project, under grant agreement No.\ 101091534.
M.G. was supported in part by the Swiss National Science Foundation under Grant 200364.

\appendix*

\section{THE CALCULI OF MULTI-LEVEL COMPUTATIONS}
\label{sec:theory}

Here we present the fundamental mathematical results that are at the basis of our framework, which was informally described in Section~\ref{sec:our_proposal}.

\subsection{Scenario and basic assumptions}

Let us consider a dynamical process of interest whose evolution can be described by a time series --- i.e., a stochastic process $X_t$ sampled at discrete times $t\in\mathbb{Z}$. 
To simplify the results presented in this paper, we adopt two assumptions. First, it is assumed that $X_t$ takes values on a discrete alphabet $\mathcal{X}$. 
Secondly, the statistics of the system of interest (which may be Markovian or non-Markovian) are assumed to be stationary --- i.e.\ to have the same statistics for different timepoints. Extensions for non-stationary processes and/or general alphabets are possible, but will be left for future work.

As for notation, random variables are in general denoted by capital letters (e.g.\ $X$, $Y$) and their realisations by lower case letters (e.g.\ $x$, $y$). 
Time series are denoted by bold letters without subscript $\bm X = \{X_t\}_{t\in\mathbb{Z}}$, finite trajectories by $\finfut{X}{t}{L}:=(X_{t},\dots,X_{t+L})$, the infinite future as $\fut{X}{t} = (X_t,X_{t+1},\dots)$, and the infinite past as $\past{X}{t} = (\dots,X_{t-1},X_t)$. For distributions of random variables, we use the random variables as subscripts, i.e., $p_X(x)=\mathbb{P}(X=x)$. We drop the subscript if it is clear from the context.
The Shannon entropy of a random variable $X$ is denoted by $H(X)$, and the mutual information between $X$ and $Y$ by $I(X;Y)$. 
Finally, we call random variables $X$ and $Y$ equivalent if and only if there is a bijection $f$ such that $X=f(Y)$ \textit{almost surely}, i.e.\ with probability 1.

\subsection{Coarse-grainings and $\epsilon$-machines}

A coarse-grained description of a variable $X$ is another variable $Z$
such that a (deterministic) coarse-graining mapping $g$ exists such that $Z=g(X)$. For example, $X$ could be a Gaussian random variable and $Z$ could be 1 if $X\geq0$ and 0 otherwise. 
The range of values that a coarse-grained representation can adopt is equal to or smaller than the range of the original system. Crucially, if $Z=g(X)$ then observing a change in $Z$ always implies that there was a change in $X$, but if $X$ is changed this does not imply that $Z$ will --- an asymmetry known as `supervenience'~\cite{supervenience2021}.

Let us extend the idea of coarse-graining from variables to time series. 
A coarse-graining of a time series $\bm X$ is another time series $\bm Z$ such that there exists a coarse-graining mapping $g$ such that $Z_t = g(\past{X}{t})$ for all $t$, so that $\bm Z$ summarises spatio-temporal trajectories of $\bm X$. 
A coarse-grained time series $\bm Z$ will be described informally as
being a \emph{macroscopic process} related to a microscopic process
$\bm X$. We denote by $\mathcal{G}_{\mathcal{X}}:=[\mathcal{X}^*\to
\mathbb{R}]$ the set of all possible coarse-graining functions taking the past $\past{X}{t}$ as input (with $\mathcal{X}^*$ denoting the set of infinite strings corresponding to all possible past trajectories), and by $\mathcal{F}_{\mathcal{X}}:=[\mathcal{X}\to \mathbb{R}]$ the subset of all purely spatial (i.e.\  non-temporal) coarse-grainings of the form $Z_t = f(X_t)$.

When considering two coarse-grainings $\bm Z^{(1)}$ and $\bm Z^{(2)}$, we use the notation $\bm Z^{(1)} \succeq \bm Z^{(2)}$ to denote when $\bm Z^{(1)}$ is by itself also a coarse-graining of $\bm Z^{(2)}$. Two coarse-grainings are equivalent if they satisfy both $\bm Z^{(1)} \succeq \bm Z^{(2)}$ and $\bm Z^{(2)} \succeq \bm Z^{(1)}$, which corresponds to cases where there exists a bijection between them. Two coarse-graining functions $g_1,g_2\in\mathcal{G}_{\mathcal{X}}$ are said to be equivalent if they lead to equivalent coarse-grainings.

\begin{lemma}\label{sec:lemma_lattice_coarsegrainings}
For a given process $\bm X$, the collection
$\mathcal{G}_{\mathcal{X}}$ with $\succeq$ forms a lattice, with the infimum given by the identity mapping and supremum given by a constant mapping. The set $\mathcal{F}_{\mathcal{X}}$ constitutes a sublattice of $\mathcal{G}_\mathcal{X}$.
\end{lemma}
\begin{proof}
Let us consider two coarse-grainings of $\bm X$ denoted by $\bm Z$ and $\bm Z'$. 
To show that $\succeq$ induces a partial ordering, one needs to prove
that if $\bm X \succeq \bm Y$ and $\bm Y \succeq \bm Z$, then $\bm X
\succeq \bm Z$, which is a natural consequence of the composition of
coarse-graining functions. The proof of the infimum and supremum
follows from the fact that $\bm X\succeq\bm Z\succeq 0$ for all
coarse-grainings $\bm Z$, where the supremmum $\bm X$ is obtained via
the identity function and $0$ is a trivial coarse-graining obtained
from a constant function mapped to $0$. Finally, the fact that both $\bm X$ and $0$ also belong to $\mathcal{F}_{\mathcal{X}}$ shows that it is a sublattice.
\end{proof}

We now introduce the idea of $\epsilon$-machines of a macroscopic processes. For this, we build on the classic definition of $\epsilon$-machine~\cite{shalizi2001causal}.

\begin{definition}[$\epsilon$-machine]
\label{def:e-machine}
The causal states of a process $\bm X$ are the equivalence classes corresponding to the following equivalence relationship (see Figure~\ref{fig:causal_states}):
\begin{equation}
\past{x}{t} \equiv_{\epsilon} \past{x}{t}' \quad \text{iff} \quad p(\finfut{x}{t+1}{L} | \past{x}{t}) = p(\finfut{x}{t+1}{L} | \past{x}{t}')~\;  \forall \finfut{x}{t+1}{L} , L\in\mathbb{N}.
\nonumber
\end{equation}
The $\epsilon$-machine of a process $\bm X$ is given by the pair $(\epsilon,T_{e,e'}^x)$, where $\epsilon(\past{X}{t}) = E_t$ is the mapping that assigns their causal state to individual trajectories, and $T_{e,e'}^x$ is the resulting transition probabilities of the form
\begin{equation}
    T_{e,e'}^x = \mathbb{P}\{ E_t=e', X_t=x | E_{t-1}=e\}~.\nonumber
\end{equation}
\end{definition}

It can be shown that the causal states are the coarsest coarse-graining of past trajectories of the process that retains full predictive power over its future evolution --- i.e.\  if the time series of causal states is denoted by $\bm E$, then $I(\past{X}{t};\finfut{X}{t+1}{L}) = I(E_t;\finfut{X}{t+1}{L})$ for all future horizons $L\in\mathbb{N}$~\cite{crutchfield1989inferring,grassberger1986toward}. Moreover, it can be shown that $\bm E$ always exhibits Markovian dynamics, irrespective of the complexity of the memory structure of $\bm X$. The $\epsilon$-machine is thought to describe the computations that give rise to the observed data representing them via an automaton~\cite{hopcroft2001introduction}, as it is the most parsimonious state machine whose transitions generate the actual data that is observed (see Section~\ref{sec:two_faces_emachines}).

By noting that the dynamics in $\bm X$ naturally induce dynamics on a coarse-graining $\bm Z$, one can directly apply Definition~\ref{def:e-machine} to $\bm Z$. 
Intuitively, the $\epsilon$-machine of a coarse-graining $\bm Z$ is an optimal representation of past trajectories at that level of resolution (i.e.\ trajectories of $\past{z}{t}$ instead of $\past{x}{t}$), which are grouped according to their forecasting abilities into causal states $\bm E'$. 
Note that the causal states of $\bm Z$ are also Markovian and informationally optimal at a macroscopic level, i.e.\ they satisfy $I(\past{Z}{t}; \finfut{Z}{t+1}{L}) = I(E'_t;\finfut{Z}{t+1}{L})$ for all $L\in\mathbb{N}$.

Let us think how the different machines relate to each other. 
If $\bm X$ represents a fully observed (physical) process governed by a given (physical) law, then the $\epsilon$-machine of $\bm X$ represents the computations taking place at this most resolved level --- representing `the laws' governing the system. 
Similarly, the $\epsilon$-machine of $\bm Z$ is the best possible reconstruction of the computations that take place at scale of resolution. Critically, such reconstruction may be losing relevant information about the underlying laws, i.e.\ about the $\epsilon$-machine of $\bm X$. 
It is natural to wonder how the collection of $\epsilon$-machines of all coarse-grainings of a given process $\bm X$ are related. Unfortunately, in general the members of such collection do not --- to the best of our knowledge --- relate to each other in a simple manner. For example, the $\epsilon$-machine of a coarse-graining of $\bm X$ may not be a coarse-graining of the $\epsilon$-machine of $\bm X$, as shown in the next example.

\begin{example}[Incompatibility between $\epsilon$-machines at different levels]
\label{ex:incompatibility_emachines}
 Consider a Markov chain $\bm X$ with three states $a$, $b$, and $c$. Let us assume that the vectors of outgoing transition probabilities of $a$ and $b$ are positive and identical. Then, the $\epsilon$-machine $\bm E$ has only two causal states, namely, $\{a,b\}$ and $c$.

 Now assume that there is no transition from $c$ to $b$, and that the transitions from $b$ to $a$ and from $c$ to $a$ have different probabilities. Suppose the coarse-graining $Z_t=f(X_t)$ is such that $Z_t=A$ if $X_t=a$ and $Z_t=B$ otherwise. It can be shown that $\bm Z$ is not Markov. Indeed, the longer $\bm Z$ remains in state $B$, the higher the probability that $X_t=c$ (as there are no transitions from $c$ to $b$). Thus, since the transition probabilities into $a$ are different for $b$ and $c$, the $\epsilon$-machine $\bm E'$ of $\bm Z$ has causal states $\{A, (BA), (BBA), (BBBA), \dots\}$. Since evidently $\bm E'$ has more than two causal states, it cannot be a coarse-graining of the $\epsilon$-machine $\bm E$ of $\bm X$.
\end{example}

\subsection{$\upsilon$-machines}

We now introduce a new class of machines associated with coarse-grainings, which we call $\upsilon$-machines --- which are illustrated in Figure~\ref{fig:e_vs_u}.

\begin{definition}[$\upsilon$-machine]
\label{def:u-machine}
The causal states of the process $\bm X$ with respect to its coarse-graining
$\bm Z$ are the collection of equivalence classes given by
\begin{equation}
\past{x}{t} \equiv_{\upsilon} \past{x}{t}' \quad \text{iff} \quad p(\fut{z}{t+1}^L | \past{x}{t}) = p(\fut{z}{t+1}^L | \past{x}{t}')~\;  \forall \fut{z}{t+1}, L\in\mathbb{N} .
\nonumber
\end{equation}
The $\upsilon$-machine of the coarse-graining $\bm Z$ is then given by the pair $(\upsilon,T_{e,e'}^x)$, where $\upsilon(\past{X}{t}) = U_t$ is the mapping that assigns the causal state with respect to $\bm Z$ to individual trajectories $\past{x}{t}$, and $T_{u,u'}^z$ is the resulting transition probabilities of the form
\begin{equation}
    T_{u,u'}^z = \mathbb{P}\{ U_t=u', Z_t=z | U_{t-1}=u\}~.\nonumber
\end{equation}
\end{definition}
Unlike the $\epsilon$-machine of $\bm X$ that predicts the whole (detailed) future $\fut{X}{t+1}$, the $\upsilon$-machine of a coarse-graining $\bm Z$ is made of the collection of causal states of the microscopic process $\bm X$ that
are needed to optimally predict only its
coarse-graining $\fut{Z}{t+1}$. 
Building on this intuition, we now show
that the causal states of the $\upsilon$-machine correspond to the coarsest coarse-graining of $\bm X$
that is maximally predictive of $\bm Z$ (which provides an alternative
characterisation of them).

\begin{proposition}\label{prop:minimality_upsilon}
    Consider $\bm Z$ a coarse-graining of $\bm X$, and $\bm U$ the time series of causal states of the $\upsilon$-machine. Then $I(U_t;\fut{Z}{t+1}^L) = I(\past{X}{t};\fut{Z}{t+1}^L)$ for all $L\in\mathbb{N}$. Moreover, if $\bm D$ is another coarse-graining of $\bm X$ such that  
    $I(D_t; \fut{Z}{t+1}^L) = I(\past{X}{t}; \fut{Z}{t+1}^L)$ for all $L\in\mathbb{N}$, then the $\upsilon$-machine of $\bm Z$ is a coarse-graining of $\bm D$.
\end{proposition}

\begin{proof}
The first assertion follows directly from the definition of the $\upsilon$-machine.
Namely, observe that
\begin{equation}
    I(\past{X}{t};\fut{Z}{t+1}^L) = \sum_{\past{x}{t}} p(\past{x}{t}) D\Big( p(\finfut{z}{t+1}{L}|\past{x}{t}) \big|\big| p(\finfut{z}{t+1}{L}) \Big).
\end{equation}
Noting that $p(\finfut{z}{t+1}{L}|\past{x}{t})$ only depends on which equivalence class $u_t$ the sequence $\past{x}{t}$ belongs to, and by the construction of $u_t,$ we can equivalently write
\begin{equation}
    I(\past{X}{t};\fut{Z}{t+1}^L) = \sum_{u_t} p(u_t) D\Big( p(\finfut{z}{t+1}{L}|u_t) \big|\big| p(\finfut{z}{t+1}{L}) \Big),
\end{equation}
and the latter is precisely equal to $I(U_t;\fut{Z}{t+1}^L).$

For the second assertion,
let us consider $\bm D$ a coarse-graining of $\bm X$, and express the mutual information between $\bm D$ and $\bm Z$ in terms of Kullback-Leibler divergences as follows:
\begin{equation}
    I(D_t; \finfut{Z}{t+1}{L}) = \sum_{d_t} p(d_t) D\Big( p(\finfut{z}{t+1}{L}|d_t) \big|\big| p(\finfut{z}{t+1}{L}) \Big).
\end{equation}
Our strategy will be to investigate under which conditions this expression is smaller than $I(\past{X}{t} ; \finfut{Z}{t+1}{L})$.

Let's first consider what happens if $D_t$ only lumps together trajectories of $\past{x}{t}$ that have the same conditional probabilities $p(\finfut{z}{t+1}{L}|\past{x}{t})$. Then, $I(D_t; \finfut{Z}{t+1}{L}) = I(\past{X}{t} ; \finfut{Z}{t+1}{L})$, as each Kullback-Leibler term in the summation stays the same, and the probabilities $p(d_t)$ just group together probabilities of the corresponding trajectories. Said differently, lumping together terms with similar conditionals results on a rearranging of the summation terms and their weights, but nothing more.

Let's now consider what happens if $D_t$ lumps at least two trajectories with different conditional probabilities $p(\finfut{z}{t+1}{L}|\past{x}{t})$, and show that in this scenario the mutual information will necessarily drop. For this, let's assume that $D_t=d_0$ is lumping together trajectories $\past{x}{t}$ and $\past{x}{t}'$ that have different conditional probabilities. Then, first notice that
\begin{align}
p(d_t)p(\finfut{z}{t+1}{L}|d_t) 
&= p(d_t,\finfut{z}{t+1}{L}) \nonumber\\
&= p(\past{x}{t},\finfut{z}{t+1}{L}) + p(\past{x}{t}',\finfut{z}{t+1}{L}) \nonumber\\
&= p(\past{x}{t})p(\finfut{z}{t+1}{L}|\past{x}{t}) + p(\past{x}{t}')p(\finfut{z}{t+1}{L}|\past{x}{t}') \nonumber
\end{align}
and then 
\begin{equation}
    p(\finfut{z}{t+1}{L}|d_t) = r\cdot p(\finfut{z}{t+1}{L}|\past{x}{t}) + s\cdot p(\finfut{z}{t+1}{L}|\past{x}{t}'),
\end{equation}
where $r+s=1$ is an affine combination. Using then the fact that the KL is strictly convex, i.e.\ that if $q_1$ and $q_2$ are different then
\begin{equation}
D(r\cdot q_1 + s\cdot q_2||p) < r\cdot D(q_1||p) + s\cdot D(q_2||p) ,
\end{equation}
one can show that the mutual information with such coarse-graining cannot be maximal, i.e.\ that some information must be lost.

This shows that any maximally informative coarse-graining can only lump trajectories that have the same conditionals. And the coarsest one is the one that lumps them all, which is by definition the $\upsilon$-machine. Therefore, any coarse-graining that is maximally informative will be compatible with the $\upsilon$-machine, but may be less coarse.
\end{proof}

It is natural to wonder how the $\upsilon$-machine and $\epsilon$-machine of a coarse-graining are related to each other. 
Our next result shows that the $\upsilon$-machine of $\bm Z$ is always more
information-powerful than the $\epsilon$-machine of $\bm Z$. Intuitively, this is a consequence of the fact that the former can rely on information that is in $\bm X$ but may not be accessible in $\bm Z$. 
\begin{lemma}\label{lemma:u_over_e}
Consider $\bm Z$ to be a coarse-graining of $\bm X$, and denote by $\bm U$ and  $\bm E'$ the time series of the causal states of its $\upsilon$-machine and $\epsilon$-machine, respectively. 
Then, the following relationship is satisfied for all future horizons $L\in\mathbb{N}$:
  \begin{align}
      I(U_t; \finfut{Z}{t+1}{L}) 
      \geq
      I(E'_t; \finfut{Z}{t+1}{L}).\nonumber
  \end{align}
\end{lemma}
\begin{proof}
By the definition of $\epsilon$-machines, we must have $I(E'_t; \finfut{Z}{t+1}{L})=I(\past{Z}{t}; \finfut{Z}{t+1}{L})$, while by the definition of $\upsilon$-machines, we have $I(U_t; \finfut{Z}{t+1}{L}) = I(\past{X}{t}; \finfut{Z}{t+1}{L})$. The lemma is then proven by a simple application of the data processing inequality, which states that since $\past{Z}{t}$ is a function (a coarse-graining) of $\past{X}{t},$ we must have $I(\past{Z}{t}; \finfut{Z}{t+1}{L}) \leq I(\past{X}{t}; \finfut{Z}{t+1}{L}).$
\end{proof}

This lemma could perhaps give the impression that the causal states of the $\epsilon$-machine of $\bm Z$ may simply be a coarse-graining of the causal states of its $\upsilon$-machine. 
Unfortunately, this is not the case ---  rather, these machines can be related in highly non-trivial ways, as illustrated by our next example.

\begin{example}[Incompatibility between $\upsilon$-machine and $\epsilon$-machine of a macroscopic process]
\label{ex:UdoesnotmaptoE}
 We reconsider the Markov chain $\bm X$ and its coarse-graining from Counterexample~\ref{ex:incompatibility_emachines}. Recall that the time series of causal states of $\bm X$, $\bm E$, alternate between two causal states states, $\{a,b\}$ and $\{c\}$. As we will show later in Proposition~\ref{lem:UisFofE}, since the coarse-graining mapping $f$ that defines $\bm Z$ is spatial, the causal states of the $\upsilon$-machine $\bm U$ of $\bm Z$ is a coarse-graining of $\bm E$; hence, it can have at most two causal states. Recall further that the time series of causal states of the $\epsilon$-machine of $\bm Z$, $\bm E'$, has causal states states $\{A, (BA), (BBA), (BBBA), \dots\}$. Since evidently $\bm E'$ has more than two causal states (i.e., more states than $\bm U$ has), there cannot be a mapping $f\in\mathcal{F}_\mathcal{E}$ such that $E'_t=f(U_t)$.
\end{example}

\subsection{Information and causal closure}

The $\upsilon$-machine of a coarse-graining $\bm Z$ captures the hidden states whose computations drive its evolution as seen from $\bm X$, while the $\epsilon$-machine of $\bm Z$ is the best effort to capture such states using only information from that scale of resolution (see Figure~\ref{fig:e_vs_u}). 
Hence, if the $\epsilon$-machine of a macroscopic process is as powerful as its $\upsilon$-machine (concretely, if there exists a bijection between the causal states of both, $E'_t$ and $U_t$), then one can say that all the relevant causal relationships take place at that macroscopic scale. This principle sets the basis for defining macroscopic processes that are \emph{causally closed}.

\begin{definition}[Causal closure]
\label{def:causal_closure}
A coarse-graining $\bm Z$ is said to be \emph{causally closed} if the coarse-grainings induced by its $\epsilon$-machine and $\upsilon$-machine are equivalent.
\end{definition}

Let us illustrate how coarse-grainings can be causally closed in a minimal example.

\begin{ex}
\label{ex:minimal_example}
Consider a stationary Markov chain $\bm X$ that can take one of three values $\mathcal{X} = \{a,b,c\}$. This system is fully described by a $3\times 3$ transition probability matrix that gives the probabilities $p(x'|x):= \mathbb{P}\{X_{t+1}=x'|X_t=x\}$ for $x,x'\in\mathcal{X}$, which has 6 degrees of freedom (as there are three conditional probabilities that sum up to 1). 
Let us consider a generic coarse-graining of such system given by $Z_t=g(X_t)$, with the mapping $g$ defined by 
\begin{equation*}
g(x)=
\begin{cases}
    0 & \text{if } x\in\{a,b\},\\
    1 & \text{if } x=c.
\end{cases}
\end{equation*}

To investigate under what conditions is $Z_t$ causally closed, let's build the $\upsilon$-machine for predicting $Z_{t+1}$ in terms of $X_t$. For this, first note that 
\begin{align*}
    \mathbb{P}\{Z_{t+1}=0 | X_t = x\} =& \,p(a|x) + p(b|x), \\
    \mathbb{P}\{Z_{t+1}=1 | X_t = x\} =& \,p(c|x).
\end{align*}
From here one can realise that the $\epsilon$-machine and $\upsilon$-machine of $\bm Z$ are equivalent if the $\upsilon$-machine is unable to distinguish the states that are merged by the coarse-graining, i.e.\ we need $\mathbb{P}\{Z_{t+1}=r | X_t = a\} = \mathbb{P}\{Z_{t+1}=r | X_t = b\}$ for $r\in\{0,1\}$. Furthermore, by noticing that $\mathbb{P}\{Z_{t+1}=1 | X_t = v\} = 1 - \mathbb{P}\{Z_{t+1}=0 | X_t = v\}$, this leads to a single condition: $p(c|a) = p(c|b)$.
\end{ex}

Let us add one further definition of a similar but seemingly weaker condition.

\begin{definition}[Information closure]
\label{def:infoclosure}
    A coarse-graining $\bm Z$ of $\bm X$ is said to be \textit{informationally closed} if it satisfies $I(\past{X}{t};\finfut{Z}{t+1}{L}|\past{Z}{t})=0$ for all $L\in\mathbb{N}$.
\end{definition}

An informationally closed coarse-graining is such that knowing the
corresponding micro-state does not provide additional information
about its future evolution over what can be obtained from the past of
the coarse-graining itself. More technically, this condition corresponds to the following Markov chain: $\past{X}{t}-\past{Z}{t}-\fut{Z}{t+1}$, where knowing $\past{Z}{t}$ makes  $\past{X}{t}$ and $\fut{Z}{t+1}$ conditionally independent. This condition makes $\past{Z}{t}$ to be a sufficient statistic of $\past{X}{t}$ for predicting $\fut{Z}{t+1}$. 
Information closure of coarse-grainings have been studied by Refs.~\cite{chang2020information,barnett2021dynamical}.

While one could intuitively think that informational closure would be a weaker notion than computational closure, our next result shows that both are, in fact, equivalent.

\begin{theorem}\label{thm:equivalence}
Let $\bm Z$ be a coarse-graining of $\bm X$. Then, the following conditions are equivalent:
\begin{itemize}
    \item[(i)] $\bm Z$ is causally closed.
    \item[(ii)] $\bm Z$ is informationally closed.
    \item[(iii)] The causal states of its $\upsilon$-machine and $\epsilon$-machine of $\bm Z$, denoted by $\bm U$ and $\bm E'$, satisfy $I(U_t;\finfut{Z}{t+1}{L}) = I(E'_t;\finfut{Z}{t+1}{L})$ for all $L\in\mathbb{N}$.
\end{itemize}

\end{theorem}

\begin{proof}
Let us first prove the equivalence between (ii) and (iii). This follows directly from the following calculation:
\begin{align}
    I(\past{X}{t};\finfut{Z}{t+1}{L}|\past{Z}{t}) 
    &= I(\past{X}{t},\past{Z}{t};\finfut{Z}{t+1}{L}) - I(\past{Z}{t};\finfut{Z}{t+1}{L}) \\
    &= I(\past{X}{t};\finfut{Z}{t+1}{L}) - I(\past{Z}{t};\finfut{Z}{t+1}{L}) \\
    &= I(U_t;\finfut{Z}{t+1}{L}) - I(E'_t;\finfut{Z}{t+1}{L})~.
\end{align}
Above, the first equality uses the chain rule for mutual information (see e.g.~\cite[Thm.2.5.2]{cover2012elements}) which states that $I(W,V;U)=I(V;U)+I(W;U|V)$ for any random variables $(U,V,W)$, the second the fact that $\bm Z$ is a coarse-graining (i.e.\ a deterministic function) of $\bm X$, and the third the properties of $\epsilon$- and $\upsilon$-machines.

Now, let's prove the equivalence between (i) and (iii). If $\bm Z$ is
causally closed then there exists a bijection between $E'_t$ and
$U_t$, and hence $I(U_t;\finfut{Z}{t+1}{L}) =
I(E'_t;\finfut{Z}{t+1}{L})$ follows directly. To prove the converse,
note that $\bm E'$ is a coarse-graining of $\bm X$, since the
causal states of a coarse-graining is a coarse-graining of the
original process. Now note that
$I(U_t;\finfut{Z}{t+1}{L}) = I(E'_t;\finfut{Z}{t+1}{L})$ implies that
$E'_t$ is maximally predictive of $\bm Z$. Due to
Proposition~\ref{prop:minimality_upsilon}, this implies that $U_t =
f(E'_t)$. This, however, implies that $\bm U$ is also a
coarse-graining of $\bm Z$. This, combined with the minimality of
$E'_t$ with respect to $\bm Z$ (i.e.\ that it is the coarses coarse-graining of $\past{Z}{t}$ that still attains optimal prediction), implies that the mapping $f$ is a bijection. 
\end{proof}

Note that Lemma~\ref{lemma:u_over_e} shows that $\upsilon$-machines are generally more powerful than the corresponding $\epsilon$-machines; the present result states that informational and causal closure are characterised by the fact of them being equally powerful.

\subsection{The lattice of spatial $\upsilon$-machines}

While the results presented in previous sections introduce useful relationships between $\upsilon$-machines and $\epsilon$-machines of coarse-grainings, the relationship between them and the $\epsilon$-machine of the process itself remains unclear. Unfortunately, no general simple relationship exists between the $\epsilon$-machine of $\bm X$ and $\upsilon$-machines of coarse-grainings for general (spatio-temporal) coarse-grainings, as shown in the following example --- which illustrates how spatio-temporal coarse-grainings can induce substantial changes into the causal organisation of the process.

\begin{example}[Incompatibility between $\upsilon$-machine and $\epsilon$-machine of the micro]
Let $\bm X$ be i.i.d., from which follows that the $\epsilon$-machine $\bm E$ is trivial. Let's focus on the following coarse-graining $Z_t=g(\past{X}{t})=X_{t-1}$. The $\upsilon$-machine can, by definition, access the entire past $\past{X}{t}$ to construct a causal state for $Z_{t+1}$. But since $Z_{t+1}=X_t$, we can set $U_t=X_t$ and obtain $H(Z_{t+1}|U_t)=0$. 
Note that such $U_t$ is minimal, since any non-trivial coarse-graing of $X_t$ will loose entropy and hence ability to predict itself. 
However, since $\bm E$ is trivial, there is no function $f\in\mathcal{F}_\mathcal{E}$ such that $U_t=f(E_t)$.
\end{example}

While the collection of all $\epsilon$-machines of coarse-grainings $\bm Z$ often don't have a clear structure linking them, our next result shows that the $\upsilon$-machines arising from spatial coarse-grainings arrange naturally into a lattice.

\begin{proposition}\label{lem:UisFofE}
  Let $f\in\mathcal{F}_\mathcal{X}$ be a spatial coarse-graining
  function and let $Z_t=f(X_t)$. Then, the causal states $U_t$ of the $\upsilon$-machine 
  of $\bm Z$, are a coarse-graining of the causal states $E_t$ of the $\epsilon$-machine of $\bm X$. Moreover, the collection of $\upsilon$-machines of all
  coarse-grainings of $\bm X$ form a lattice,
  which is homomorphic to the
  lattice of coarse-grainings of $\bm X$.
 \end{proposition}

\begin{proof}
Let us first show that there exists a spatial coarse-graining $f'\in\mathcal{F}_\mathcal{X}$ such that $U_t=f'(E_t)$. Recall that $E_t$ corresponds to the equivalence class of trajectories $\past{X}{t}$ that have the same conditional distribution $p(\finfut{x}{t+1}{L}|e_t)$, and that $U_t$ corresponds to the equivalence class of trajectories $\past{X}{t}$ that have the same conditional distribution $p(\finfut{z}{t+1}{L}|u_t)$. Now, by considering a given future trajectory $\finfut{z}{t+1}{L}$, one sees  directly that
\begin{equation}
    p(\finfut{z}{t+1}{L}|\past{x}{t}) = \sum_{\finfut{x}{t+1}{L} \in f^{-1}(\fut{z}{t+1})} p\big(\finfut{x}{t+1}{L}|\past{x}{t}\big).
\end{equation}
Therefore, if $\past{x}{t}$ and $\past{x}{t}'$ are two trajectories
such that $p(\finfut{x}{t+1}{L}|\past{x}{t}) =
p(\finfut{x}{t+1}{L}|\past{x}{t}')$ for all $\finfut{x}{t+1}{L}$ and
$L\in\mathbb{N}$, then also $p(\finfut{z}{t+1}{L}|\past{x}{t}) = p(\finfut{z}{t+1}{L}|\past{x}{t}')$. This implies that the coarse-graining that leads to $E_t$ contains all distinctions made by the coarse-graining that leads to $U_t$, and perhaps some more --- the ones needed to predict $\finfut{X}{t+1}{L}$ rather than just $\finfut{Z}{t+1}{L}$. This shows that $U_t$ necessarily is a coarse-graining of $E_t$.

From the previous argument, it is clear that the collection of all
$\upsilon$-machines of spatial coarse-grainings inherit the
partial-ordering relationship of their corresponding coarse-grainings.
Furthermore, it is direct to see that the $\upsilon$-machine of the
coarse-graining function $f(x)=x$ is the finest of all, and the one of
$f(x)=0$ is the coarsest. This proves that the partial ordering over
$\upsilon$-machines generates a lattice, and the fact that the mapping
from coarse-grainings $\bm Z$ to $\upsilon$-machines may not
be one-to-one makes it an homomorphism.
\end{proof}

This result formally shows that, at least for spatial coarse-grainings, the $\upsilon$-machine corresponds to the part of the underlying $\epsilon$-machine that plays a role in running the macroscopic process $\bm Z$.

\subsection{Computational closure and the structure of emerging scales}

Let us now consider the notion of closure from a computational
perspective. For this, let us focus on a deterministic automaton (not necessarily related to an $\epsilon$-machine) with
states in the set $\mathcal{E}$ and inputs in the alphabet $\mathcal{X}$, so that a tuple $(e_t,x_{t+1})$ with $e_t\in\mathcal{E}$ and $x_{t+1}\in\mathcal{X}$ uniquely determine the next state $e_{t+1}\in\mathcal{E}$~\cite{hopcroft2001introduction}. Then, for a given initial state $e_0$, an input sequence of symbols $x_1,\dots,x_T$ generate a corresponding sequence of states $e_0,e_1,\dots,e_T$. Let us consider what happens if one has sequences not of symbols in $\mathcal{X}$, but from a coarse-grained alphabet $\mathcal{Z}=\{f(x): x\in\mathcal{X}\}$ established by a spatial coarse-graining mapping $z=f(x)$, $f\in\mathcal{F}_{\mathcal{X}}$. In general, the symbols in $\mathcal{Z}$ may not be sufficient to uniquely specify the transitions between states in $\mathcal{E}$, so that the states $\mathcal{E}$ do not constitute a deterministic automaton under inputs in $\mathcal{Z}$. If it does, we say that the coarse-graining is \emph{computationally closed}. 
We formalise these ideas in the next definition.

\begin{definition}[Computational closure]
\label{def:compu_closure}
For a given deterministic automaton with states in $\mathcal{E}$ and inputs in $\mathcal{X}$, the coarse-grainings of inputs $f$ and of states $f^*$ are said to be \emph{computationally closed} if the resulting transitions between coarse-grained states $\mathcal{E}'=\{E'=f^*(E):E\in\mathcal{E}\}$ is also a deterministic automaton under inputs in $\mathcal{Z}=\{z=f(x):x\in\mathcal{X}\}$.
\end{definition}

Put simply, computational closure refers to a pair of coarse-grainings $(f,f^*)$ such that the coarse-graining of the state is `compatible' with the coarse-graining of the input --- in the sense that they give rise to a new discrete automaton. 

The notion of computational closure can be naturally applied to macroscopic coarse-grainings, by considering how their $\epsilon$-machines --- which can be seen as a deterministic automata (see Section~\ref{sec:two_faces_emachines}) --- relates to the $\epsilon$-machine of the underlying microscopic process. Following Definition~\ref{def:compu_closure}, one can say that 
a macroscopic process $Z_t=g(\past{X}{t})$ is computationally closed if the time series $\bm E'$ of causal states of its $\epsilon$-machine is computationally closed with respect to the time series of causal states $\bm E$ of the $\epsilon$-machine of the microscopic process $\bm X$.
In other words, $Z_t=g(\past{X}{t})$ is
computationally closed if and only if there is a spatial coarse-graining $f^*$ such that $E'_t=f^*(E_t)$. 

While computational closure seems to be a stronger condition than causal closure, our next example (adapted from{~\cite[Example~6.1]{Pfante_LevelID}}) shows this is not the case.

\begin{example}[Computational closure does not imply information closure, adapted from{~\cite[Example~6.1]{Pfante_LevelID}}]
\label{ex:compcl_no_infocl}
Consider the following time-homogeneous Markov chain $\bm X$ with transition probability matrix
    \begin{equation}
    P=\left[\begin{array}{cc|cc}
        0.5 & 0.5 & 0& 0 \\
        0 & 0 & 0.5 & 0.5\\
        \hline
        0.5 & 0.5 & 0 & 0 \\
        0 & 0 & 0.5 & 0.5 \\
    \end{array}\right],
\end{equation}
where, following standard convention, the entry in row $i$ and column $j$ denotes $\mathbb{P}\{X_{t+1}=j|X_t=i\}$, i.e., transition probability matrices are row stochastic.
The lines indicate that the coarse-graining function $f$ merges states 1 and 2 to $A$ and states 3 and 4 to $B$. If the starting distribution $p_{X_0}$ coincides with the uniform distribution on $\{1,2,3,4\}$, then the resulting coarse-graining $\bm Z$ is i.i.d. Hence, its $\epsilon$-machine is a trivial one, and hence is a coarse-graining of the $\epsilon$-machine of $\bm X$. 
However, this coarse-graining is not causally closed. In effect, by knowing $X_t$ one can perfectly predict $Z_{t+1}.$ Namely, writing out explicitly, one can find for this particular Markov chain that
\begin{align}
    p(Z_{t+1}=A | X_t = x_t ) & = \left\{ \begin{array}{ll} 1 &\mbox{if } x_t=1 \mbox{ or } x_t=4,\\ 0 & \mbox{otherwise,} \end{array} \right.\label{ex:compcl_no_infocl:eqn1}
\end{align}
and likewise for coarse state B.
Hence the $\upsilon$-machine of $\bm Z$ is non-trivial. Similarly, $\bm Z$ is not informationally closed since 
\begin{align}
    I(\past{X}{t}; Z_{t+1} | \past{Z}{t}) &= H(Z_{t+1}| \past{Z}{t}) -  H(Z_{t+1}|\past{X}{t},\past{Z}{t}) \nonumber \\
    & = H(Z_{t+1})  = 1.
\end{align}
To establish this, we observe that since $\bm Z$ is i.i.d., we have $H(Z_{t+1}| \past{Z}{t})=H(Z_{t+1})$; due to~\Cref{ex:compcl_no_infocl:eqn1}, we have $H(Z_{t+1}|\past{X}{t},\past{Z}{t})=0;$
and the last equality follows since the marginal distribution of $Z_t$ is uniform in our example.
\end{example}

With all in place, we have all the pieces to put forward our next result (illustrated by Figure~\ref{fig:megalattice}).

\begin{theorem}\label{teo:compu_closure}
Spatial coarse-grainings that are causally closed are also
computationally closed. Moreover, the set of all computationally
closed coarse-grainings form a sub-lattice, which maps (possibly not
one-to-one) into a lattice of coarse-grainings of $\epsilon$-machines
via an homomorphism. Furthermore, when restricted to this collection, the following diagram commutes:

\begin{center}
  \if1\compiletikz
  \includetikz{tikz/}{CommutingDiagram}
  \else
    \includegraphics{CommutingDiagram.pdf}
  \fi 
\end{center}

\noindent
where $f,f^*$ are coarse-graining mappings, and $\epsilon,\epsilon'$ are operators that goes from stochastic processes to their
$\epsilon$-machine. 
\end{theorem}

\begin{proof}
    To prove the first part of the theorem, let's consider a causally closed macroscopic process $\bm Z$ which is generated by a spatial coarse-graining --- i.e. $Z_t=f(X_t)$. 
    Proposition~\ref{lem:UisFofE} guarantees that, because of coming from a spatial coarse-graining, the causal states of the $\upsilon$-machine of $\bm Z$ are a coarse-graining of the causal states of $\bm X$. 
    Additionally, due to the definition of causal closure, the $\epsilon$-machine of $\bm Z$ is equal to its $\upsilon$-machine. Combining these two facts one can conclude that the causal states of the $\epsilon$-machine of $\bm Z$ must be a coarse-graining of the causal states of the $\epsilon$-machine of $\bm X$, proving that $\bm Z$ is computaionally closed.

    Let us prove now that all computationally closed coarse-grainings for a lattice, and that their $\epsilon$-machines also do. 
    All causally closed spatial coarse-grainings are a
    sub-lattice of $\mathcal{F}_\mathcal{X}$. 
    Furthermore, as causal closure in this case implies computational closure, then the collection of all $\epsilon$-machines of those are also sub-lattice of the set of all possible spatial
    coarse-grainings of the $\epsilon$-machine of $\bm X$. In virtue
    of this, the operation of calculating the $\epsilon$-machine of a
    causally closed coarse-graining, denoted by $\epsilon'(\cdot)$, establishes a natural map between coarse-grainings of $\bm X$ and coarse-grainings of its $\epsilon$-machines.

  To conclude, let's prove the conmutation between coarse-graining $f$ and calculation of causal states $\epsilon(\cdot)$.  For this, let's note that $\epsilon(\cdot)$ preserves the partial ordering relationship of coarse-grainings of $\bm X$ --- indeed, if the coarse-graining of $\bm X$ $f$ is finer than $f'$, then the $\upsilon$-machine of $Z_t=f(X_t)$ is finer than the one of $Z'_t=f'(X_t)$. This fact, combined with the identity between $\upsilon$-machines and $\epsilon$-machines that holds due to causal closure, concludes the proof.
\end{proof}

We note that the requirement of focusing on spatial coarse-grainings is essential, as the implication from information closure to computational closure can be broken in the case of general spatio-temporal coarse-grainings --- as illustrated by our next example.

\begin{example}[Causal closure doesn't always imply computational closure]
    Let $X_t$ be a Markov chain such that its causal states $E_t$ satisfy $E_t=X_t$. 
    In this scenario, let's consider the coarse-graining $Z_t=g(\past{X}{t})=X_{t-1}$. A direct calculation shows that
    \begin{align}
    I(\finfut{Z}{t+1}{L};\past{X}{t}|\past{Z}{t}) 
    &= H(\finfut{Z}{t+1}{L}|\past{X}{t}) - H(\finfut{Z}{t+1}{L}|\past{Z}{t})\nonumber\\
    &= H(\finfut{X}{t}{L}|\past{X}{t}) - H(\finfut{X}{t}{L}|\past{X}{t-1})\nonumber\\
    &=0,
    \end{align}
    which implies that $Z_t$ is informationally closed, and hence --- due to Theorem~\ref{teo:compu_closure} --- also causally closed. 
    Note that $\bm Z$ is a Markov chain, and $E’_t=X_{t-1}$. But then $E’_t=E_{t-1}$, so that in general there no spatial coarse-graining mapping $E_t$ to $E'_t$.
\end{example}

Let us now turn towards characterising the collections of coarse-grainings that lead to the same $\epsilon$-machine, corresponding to the pre-images of the mapping $\phi'$.

\begin{definition}\label{def:comp_equivalence}
Two coarse-grainings $\bm Z$ and $\bm Z'$ are said to be computationally equivalent if their $\epsilon$-machines are equivalent.
\end{definition}

Let us illustrate how coarse-grainings can be computationally equivalent in a minimal example.

\begin{ex}
Let us revisit the setting of Example~\ref{ex:minimal_example}.
To investigate under what conditions is $Z_t$ causally closed, let's build the $\upsilon$-machine for predicting $Z_{t+1}$ in terms of $X_t$. For this, first note that 
\begin{align*}
    \mathbb{P}\{Z_{t+1}=0 | X_t = x\} =& \,p(a|x) + p(b|x), \\
    \mathbb{P}\{Z_{t+1}=1 | X_t = x\} =& \,p(c|x).
\end{align*}
From here one can realise that the $\epsilon$-machine and $\upsilon$-machine of $\bm Z$ are equivalent if the $\upsilon$-machine is unable to distinguish the states that are merged by the coarse-graining, i.e.\ we need $\mathbb{P}\{Z_{t+1}=r | X_t = a\} = \mathbb{P}\{Z_{t+1}=r | X_t = b\}$ for $r\in\{0,1\}$. Furthermore, by noticing that $\mathbb{P}\{Z_{t+1}=1 | X_t = v\} = 1 - \mathbb{P}\{Z_{t+1}=0 | X_t = v\}$, this leads to a single condition:
\begin{align*}
p(c|a) &= p(c|b).
\end{align*}
Furthermore, a direct calculation shows that 
\begin{align*}
    \mathbb{P}\{Z_{t+1}=r | Z_t=s\} 
    &= \frac{ \mathbb{P}\{Z_{t+1}=r, Z_t=s\} }{\mathbb{P}\{Z_t=s\} } \\
    &= \sum_{u\in g^{-1}(r)} \sum_{v\in g^{-1}(s)}
    \xi(v) p(u|v), 
\end{align*}
where $\xi(v)$ are non-negative weights with $\sum_{v\in g^{-1}(s)} g(v) = 1$ given by
\begin{equation}
    \xi(v) = \frac{ \mathbb{P}\{ X_t = v \} }
    { \sum_{v'\in g^{-1}(v)} \mathbb{P}\{ X_t = v' \} }.
\end{equation}
This implies that
\begin{align}
    \mathbb{P}\{ Z_{t+1}=1 | Z_t = 0\} &= p(c|a),\\
    \mathbb{P}\{ Z_{t+1}=1 | Z_t = 1\} &= p(c|c).
\end{align}
Clearly, if $p(c|a) \neq 1/2$ or $p(c|c) \neq 1/2$, then the $Z_t$ is not independently distributed and hence its $\epsilon$-machine is isomorphic to itself.

From this analysis, one can see that there are two types of scenarios in where $Z_t$ is causally closed. One is whether $Z_t$ carries non-trivial computations or not (i.e.\  if $Z_t$ is i.i.d.); the other is if $Z_t$ carries different computations than $X_t$ or not (equivalence happens when $p(x|a) = p(x|b)$ for $x\in\{a,b,c\}$, which leads to the causal states of $X_t$ to be equal than $Z_t$). It is direct to confirm that the four possible combinations of these two binary possibilities give rise to the four scenarios illustrated in Figure~\ref{fig:minimal_example}.
\end{ex}

Let us define $\Delta_L(\bm X,\bm Z) = I(\past{X}{t};\finfut{X}{t+1}{L}) - I(\past{Z}{t};\finfut{Z}{t+1}{L})$, which is the information that is being processed in $\bm X$ but not in $\bm Z$ within a time horizon $L$. If $\bm Z$ a coarse-graining of $\bm X$, then $\Delta_L(\bm X,\bm Z)\geq 0$ for all $L$ due to the data-processing inequality. 
Note, for example, that if $\bm Z$ is a spatial coarse-graining of $\bm X$ and the statistics of both are i.i.d., then they are computationally equivalent and also satisfy $\Delta_L(\bm X,\bm Z)=0$. Our next result generalises this relationship.

\begin{lemma}
    If a computationally closed spatial coarse-graining $Z_t=f(X_t)$ satisfies
    $\Delta(\bm X,\bm Z)= 0$, then $\bm Z$ and $\bm X$ are computationally equivalent.
\end{lemma}

\begin{proof}
Note that $\Delta(\bm X,\bm Z) = 0$ if and only if $I(E_t;\fut{X}{t+1}) = I(E'_t;\fut{Z}{t+1})$, where $E_t$ and $E'_t$ are the causal states of $\bm X$ and $\bm Z$. 
Also, because of computational closure, there is a function $f^*$ such that $E'_t=f^*(E_t)$.

If $\bm X$ and $\bm Z$ are not computationally equivalent, then $f^*$
is not invertible. If that is the case, then the definition of causal
states implies that there exists a trajectory $\finfut{X}{t+1}{L}$
such that $I(\finfut{X}{t+1}{L}; E'_t) < I(\finfut{X}{t+1}{L}; E_t)$,
as $E_t$ is by definition the coarsest coarse-graining that attains
$I(\finfut{X}{t+1}{L}; E_t) = I(\finfut{X}{t+1}{L};\past{X}{t})$ and thus
$E'_t$ needs to be  coarser than $E_t$. The data processing inequality then shows that
\begin{equation}
    I(\finfut{Z}{t+1}{L}; E'_t) \leq 
    I(\finfut{X}{t+1}{L}; E'_t) < 
    I(\finfut{X}{t+1}{L}; E_t),
\end{equation}
proving the desired result.
\end{proof}

We conclude this section studying how sequences of coarse-grainings decompose the total information processed by the system.

\begin{proposition}
\label{pro:info_decomposition}
Consider $\mathcal{Z} = \{\bm Z^i\}_{i=1}^m$ a sequence of computationally closed coarse grainings with $\bm Z^1=\bm X$ and $\bm Z^m=\bm 1$ the constant coarse-graining. Then
\begin{equation}
    I(\past{X}{t} ; \fut{X}{t+1} ) = \sum_{i=1}^{m-1} \Delta(\bm Z^k, \bm Z^{k+1}).
\end{equation}
\end{proposition}
\begin{proof}
A direct calculation shows that
\begin{align}
    I(\past{X}{t} ; \fut{X}{t+1} ) 
    &= I(\past{X}{t} ; \fut{Z}{t+1}^1, \fut{Z}{t+1}^2, ,\dots,\fut{Z}{t+1}^m ) \nonumber\\
    &= \sum_{k=1}^{m-1} I(\past{X}{t} ; \fut{Z}{t+1}^k| \fut{Z}{t+1}^{k+2},\dots,\fut{Z}{t+1}^m )\nonumber\\
    &= \sum_{i=1}^{m-1} I(\past{X}{t} ; \fut{Z}{t+1}^k| \fut{Z}{t+1}^{k+1})~.
\end{align}
Then, the results follows from noting that
\begin{align}
    I(\past{X}{t} ; \fut{Z}{t+1}^k| \fut{Z}{t+1}^{k+1} )
    &= I(\past{X}{t} ; \fut{Z}{t+1}^{k+1}) - I(\past{X}{t} ; \fut{Z}{t+1}^k ) \nonumber\\
    &= I(\past{Z}{t}^{k+1} ; \fut{Z}{t+1}^{k+1}) - I(\past{Z}{t}^k ; \fut{Z}{t+1}^k ),\nonumber
\end{align}
where the last equality uses the fact that each coarse-gaining is informationally closed.
\end{proof}

This result shows that a computationally closed coarse-graining $\bm Z$ that is not computationally equivalent to $\bm X$ decomposes the information processed by $\bm X$ into what is processed in the macro-scale (i.e.\ by $\bm Z$, quantified by $I(\past{Z}{t}; \fut{Z}{t+1}$) and what goes below (in $\bm X$ but not $\bm Z$, quantified by $\Delta(\bm X,\bm Z)$).

\subsection{Closure and lumpability}
\label{sec:lumpability}

\begin{figure*}
\includegraphics[width=\textwidth]{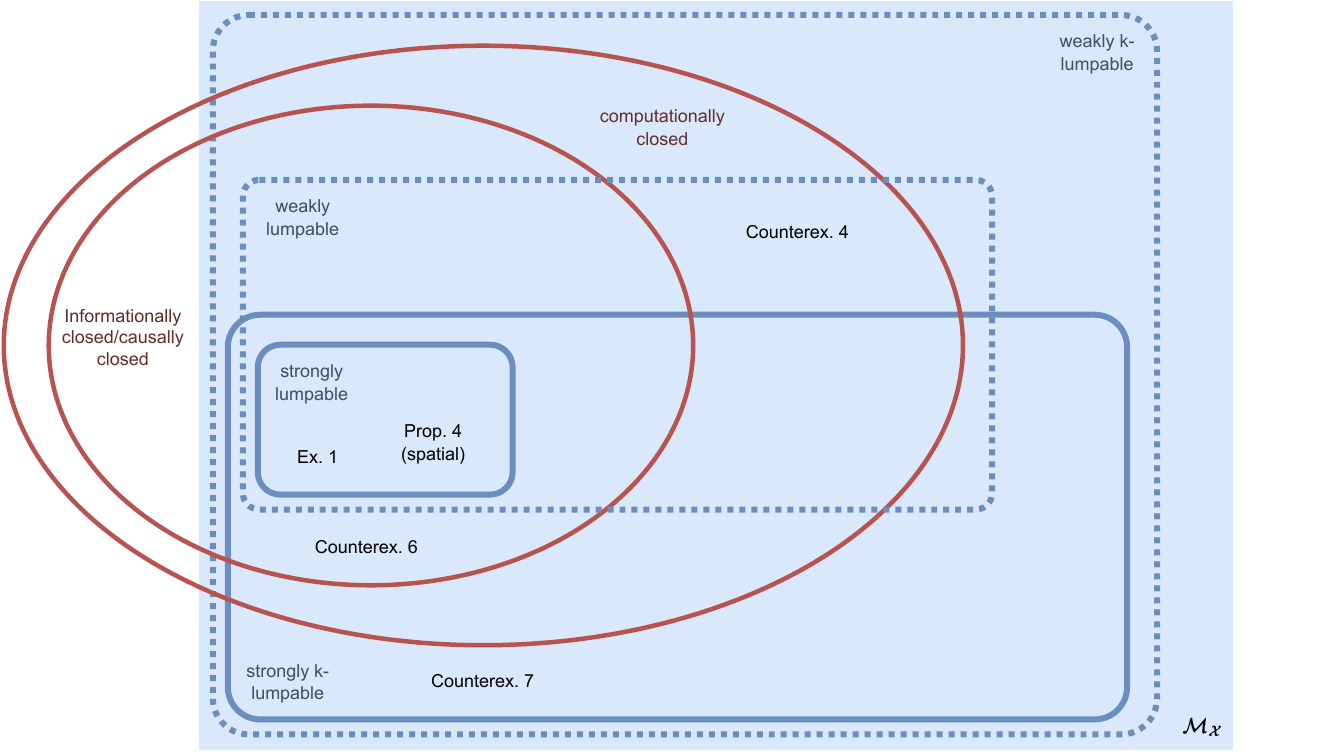}
\caption{Diagram illustrating the relationships between closure and lumpability of Markov chains. Informational/causal closure imply computational closure (Theorem~\ref{teo:compu_closure}). Within the space of Markov $\bm X$, strong lumpability of $\bm X$ implies information closure (Proposition~\ref{prop:Markov_Lumpability}). The same does not hold for weak lumpability and computational closure: If $\bm X$ is weakly lumpable, then the same does not need to hold for $\bm E$ due to the minimality property of $\varepsilon$-machines. The diagram refers to (counter)examples in the text. Indeed, Example~\ref{ex:minimal_example} is strongly lumpable, while Counterexample~\ref{ex:compcl_no_infocl} is weakly lumpable.}
\label{fig:overview}
\end{figure*}

Here we explore the relationship between informational, causal, and computational closure and `lumpability' --- see Figure~\ref{fig:overview} for an overview of the results we will develop in this subsection.
Loosely speaking, lumpability refers to Markov chains that have
coarse-grainings that also exhibit Markov dynamics.  Markedly, one
should note that lumpability is the exception and not the rule~\cite[Th.~31]{GurvitsLedoux_MarkovPropertyLinearAlgebraApproach}, as most coarse-graininings break the Markov property of the base system. 
\begin{ex}
Let's consider $\bm X$ where $X_t$ can take one of four possible values, denoted as $a,b,c$, and $d$. Consider the random initial condition where $X_0$ can take any value with equal probability, and the following Markov dynamics: if $X_t\in\{a,b\}$ then $X_{t+1}$ is equal to $c$ or $d$ with equal probability, and if $X_t\in\{c,d\}$ then $X_{t+1}$ is equal to $a$ or $b$ with equal probability. 
Now, consider a coarse-graining $\bm Z$ determined by the following
coarse-graining function: $g(a)=g(b)=A$ and $g(c)=g(d)=B$.
It can be seen than $\bm Z$ has deterministic dynamics, as $p(B|A) = p(A|B) = 1$. Furthermore, one can see that there are two causal states: one of all sequences that finish with $A$, and another with all sequences that finish in $B$. 
Interestingly, one can check that this $\epsilon$-machine is
isomorphic to the one of $\bm X$, and capture the same amount of
future information (1~bit). So, $\bm Z$ is essentially a more compact
representation than $X$ of the same computational process.
\end{ex}

Building on these ideas, let us start by introducing the formal definition of lumpability~\cite[\S6.3 \& \S6.4]{Kemeny_FMC} --- while acknowledging that other flavours of lumpability exist~\footnote{There are also other notions of lumpability, such as exact lumpability~\cite{buchholz_1994}, and higher-order lumpability~\cite{GurvitsLedoux_MarkovPropertyLinearAlgebraApproach} (where the coarse-graining $\bm Z$ is a Markov chain of higher order).}.

\begin{definition}[Strong and weak lumpability]\label{def:lumpability}
    A time-homogeneous Markov chain $(W_0,W_1,\dots)$ is \emph{strongly lumpable} w.r.t.\ a spatial coarse-graining function $f$ if for every initial distribution $p_{W_0}$ the resulting coarse-graining $Z_t=f(W_t)$ is a Markov chain, and its transition probabilities do not depend on the choice of $p_{W_0}$. A time-homogeneous Markov chain is \emph{weakly lumpable} if $\bm Z$ is Markov for at least one initial distribution $p_{W_0}$.
\end{definition}

Evidently, if a Markov chain is strongly lumpable, then it is also weakly lumpable. A useful information-theoretic characterisation of lumpability is given as follows. 
Thanks to the data-processing inequality, the following relationships always hold if $\bm X$ is Markovian:
\begin{align}
H(Z_t|Z_{t-1}) 
\geq H(Z_t|\past{Z}{t-1}) 
\geq H(Z_t|\past{X}{t-1}) 
= H(Z_t| X_{t-1}) ~.\nonumber
\end{align}
As shown by the next lemma, strong lumpability is equivalent to the collapse of the two inequalities into an equality, while weak lumpability is just the collapse of the first one.

\begin{lemma}\label{lemma:strong_lump_char}
    If $\bm X$ is a stationary, irreducible and aperiodic Markov chain and if $f$ is a coarse-graining function defined on the state space $\mathcal{X}$ of $\bm X$, then strong lumpability is equivalent to
    \begin{equation}\label{eq:lumpy1}
        H(Z_{t+1}|Z_{t})=H(Z_{t+1}|X_t) \quad \text{for all $t$}
    \end{equation}
    where $Z_t=f(X_t)$.
    In contrast, weak lumpability is equivalent to
    \begin{equation}\label{eq:lumpy2}
        H(Z_{t+1}|Z_t)=H(Z_{t+1}|\past{Z}{t}) \quad \text{for all $t$.}
    \end{equation}
\end{lemma}
\begin{proof}
    See Ref.~\cite[Th.~9 \& Prop.~6]{GeigerTemmel_kLump}.
\end{proof}
Note that this results applies to the stationary setting, i.e., in the context of Definition~\ref{def:lumpability}, the initial distribution $p_{X_0}$ must coincide with an invariant distribution (which is unique in the irreducible and aperiodic case). 

This result can be understood as follows. 
When trying to estimate $Z_{t+1}$ from its past, one can consider this
a two-step process: first one estimates $X_t$ from $\past{Z}{t}$, and
then one estimates $Z_{t+1}$ 
from  $X_t$. Strong and weak
lumpability are connected to these two steps. Indeed, Counterexample~\ref{ex:compcl_no_infocl} shows that, by carefully
choosing the initial distribution $p_{X_0}$, one can ensure that no
matter the value of  $Z_t, Z_{t-1}, Z_{t-2}, \dots$, one always ends
up with the same conditional distribution $p(X_t|\past{Z}{t})$. Thus, knowing more of the past $\past{Z}{t}$ than $Z_t$ does not help to obtain a better estimate of $X_t$. Hence, even if the conditional distribution of $Z_{t+1}$ given $X_t$ differs for all $X_t$, knowing more than just $Z_t$ does not allows us to make a better estimate of $Z_{t+1}$. In contrast, in strong lumpability the distribution of
$Z_{t+1}$ is the same for all values of $X_t$ that are compatible with
$Z_t$. Therefore, while knowing more of the past $\past{Z}{t}$ than
$Z_t$ may give us a better estimate of the current $X_t$, this
additional knowledge cannot improve our estimate of $Z_{t+1}$.

Our first result shows that, if a microscopic process $\bm X$ has Markov dynamics and is strongly lumpable w.r.t.\ the coarse graining function $f$, then the resulting macroscopic process $Z_t=f(X_t)$ is informationally closed.

\begin{proposition}\label{prop:Markov_Lumpability}
    Let $\bm X$ be a stationary, irreducible, and aperiodic Markov chain, and suppose that it is strongly lumpable w.r.t.\ the coarse graining function $f\in\mathcal{F}_{\mathcal{X}}$. Then, the process $Z_t=f(X_t)$ is informationally closed. Further, if a coarse-graining $\bm Z$ is informationally closed and Markov, then $\bm X$ is strongly lumpable w.r.t.\ the coarse-graining function $f$.
\end{proposition}

\begin{proof}
    Using Lemma~\ref{lemma:strong_lump_char}, a direct derivation shows that
    \begin{align}
        &I(\finfut{Z}{t+1}{L};\past{X}{t}|\past{Z}{t})\notag\\
        &=H(\finfut{Z}{t+1}{L}|\past{Z}{t})-H(\finfut{Z}{t+1}{L}|\past{X}{t},\past{Z}{t})\\
        &\stackrel{(a)}{=}H(\finfut{Z}{t+1}{L}|Z_t)-H(\finfut{Z}{t+1}{L}|X_t) \label{eq:proof:eq1}\\
        &\stackrel{(b)}{=}\sum_{\ell=t+1}^{t+L} \Big( H(Z_\ell|Z_{\ell-1})-H(Z_\ell|\finfut{Z}{t+1}{\ell-1},X_t)\Big)\\
        &\stackrel{(c)}{\le} \sum_{\ell=t+1}^{t+L} \Big(H(Z_\ell|Z_{\ell-1})-H(Z_\ell|X_{\ell-1}) \Big)\\
        &\stackrel{(d)}{=}0\,.
    \end{align}
    Above, $(a)$ follows because $Z_t=f(X_t)$ and because both $\bm X$
    and $\bm Z$ are Markov; $(b)$ from the chain rule of entropy and
    Markovianity of $\bm Z$; 
    $(c)$ from the data processing inequality, stating that
    $X_{\ell-1}$ contains at least as much information about $Z_\ell$
    as $(Z_{t+1}^{\ell-1},X_t)$ does; and $(d)$ from stationarity and Eq.~\eqref{eq:lumpy1}. 
    
    For the second part of the proof, note that information closure and Markovianity of $\bm Z$ imply that~\eqref{eq:proof:eq1} holds for all $L$. Instantiating this for $L=1$ yields Eq.~\eqref{eq:lumpy1}, which completes the proof with~\cite[Th.~9]{GeigerTemmel_kLump}.
\end{proof}

Please note that the above lemma is only valid for stationary, irreducible, and aperiodic Markov chains. The reason for this is that the equivalence between strong lumpability and Eq.~\eqref{eq:lumpy1} is only guaranteed for this class of Markov chains~\footnote{In essence, strong lumpability requires an equality of certain probability distributions (cf.~\cite[Th.~6.3.2]{Kemeny_FMC}); this equality of distributions is implied by the equality $H(Z_t|Z_{t-1})=H(Z_t|X_{t-1})$ only if the conditioning event $X_{t-1}$ covers the whole state space $\mathcal{X}$ with positive probability. This is ensured for irreducible Markov chains with finite $\mathcal{X}$~\cite[Th.~4.1.4]{Kemeny_FMC}. Whether the requirement of aperiodicity can be removed, or whether the requirement of stationarity can be relaxed, is currently unclear.}. 
Also note that if the coarse-graining function $f$ does not produce a first-order, but a higher-order Markov chain, the connection between (higher-order) lumpability and informational closure can break down --- as illustrated by our next two examples.

\begin{example}[Lumpability is not implied by information closure and high-order Markovianity]
Let's consider $\bm X$ where $X_t$ can take one of four possible values, denoted as $a,b,c$, and $d$, and a coarse-graining $\bm Z$ determined by the following coarse-graining function: $g(a)=g(b)=A$ and $g(c)=g(d)=B$. Suppose that $X_0$ is uniformly distributed and that we have fully
deterministic dynamics where $a\rightarrow b\rightarrow c\rightarrow
d\rightarrow a$. Then, one can check that $\bm Z$ becomes a deterministic Markov process with memory 2 (indeed, it can be shown that $\bm X$ is strongly 2-lumpable), and its $\epsilon$-machine is given by four causal states: one for all sequences that finish with $AA$, another for $AB$, other for $BA$, and the last for $BB$. It can be shown that $\bm Z$ is informationally closed, although $\bm X$ with $f$ is not strongly (first-order) lumpable.
\end{example}

\begin{example}[High-order strong lumpability does not imply information closure]\label{ex:higherorder}
Let $\bm X$ be the Markov process from~\cite[Example
15]{GeigerTemmel_kLump}. Specifically, let the transition probability
matrix of $\bm X$ be given as 
\begin{equation}
    P=\left[\begin{array}{cc|cc}
        0.6 & 0.4 & 0& 0 \\
        0.3 & 0.2 & 0.1 & 0.4\\
        \hline
        0.2 & 0.05 & 0.375 & 0.375 \\
        0.2 & 0.05 & 0.375 & 0.375 \\
    \end{array}\right].
\end{equation}
The lines indicate that the coarse graining-function $g$ merges states 1 and 2 to $A$ and states 3 and 4 to $B$. It can be shown that the resulting coarse-graining $\bm Z$, $Z_t=g(X_t)$ is a second-order Markov chain, and that the coarse-graining is strongly 2-lumpable. A coarse graining is strongly 2-lumpable if and only if we have~\cite[Th.~9]{GeigerTemmel_kLump}
\begin{equation}\label{eq:lump_not_IC}
    H(Z_{t+1}|Z_t,X_{t-1}) = H(Z_{t+1}|Z_t,Z_{t-1}) \approx 0.733.
\end{equation}
Informational closure, for a Markov chain $\bm X$ and a second-order
Markov chain $\bm Z$, requires that $I(\finfut{Z}{t+1}{L};\past{X}{t}|\past{Z}{t})=0$ for every $L$. However,
\begin{multline}
    I(\finfut{Z}{t+1}{L};\past{X}{t}|\past{Z}{t}) \ge 
    I({Z}_{t+1};\past{X}{t}|\past{Z}{t})\\
    = H(Z_{t+1}|Z_t,Z_{t-1})-H(Z_{t+1}|X_t).
\end{multline}
A computation of the latter quantity yields $H(Z_{t+1}|X_t)\approx 0.552$, which with~\eqref{eq:lump_not_IC} leads to $I(\finfut{Z}{t+1}{L};\past{X}{t}|\past{Z}{t}) \ge 0.18$, indicating that informational closure is ruled out.
\end{example}

Additionally, it can be shown that Markovianity of the coarse-graining $\bm Z$ is in general not sufficient for information closure to hold~\cite{Pfante_LevelID}. Indeed, if the Markov chain $\bm X$ is only weakly (but not strongly) lumpable w.r.t.\ the coarse graining function $f$, then $\bm Z$ may be a Markov chain but may fail to be informationally closed, as illustrated in Counterexample~\ref{ex:compcl_no_infocl} above.

We now establish a general relationship between strong lumpability and information closure on the one hand, and weak lumpabiliy and computational closure in the other. 
We knew from Theorem~\ref{teo:compu_closure} that information closure is stronger than computational closure, and weak and strong lumpability provide a complementary angle to contemplate this relationship. The following result builds on Proposition~\ref{prop:Markov_Lumpability} and the fact that while many processes $\bm X$ may not be Markovian, the dynamics of their $\epsilon$-machines always are.

\begin{theorem}\label{teo:lumpability}
    Let $\bm E$ be the $\epsilon$-machine of a stationary process $\bm X$. Then, the following statements hold true:
    \begin{enumerate}
        \item If $\bm Z$ is a computationally closed coarse-graining of $\bm X$, then $\bm E$ is weakly lumpable.
        \item If $\bm E$ is weakly lumpable w.r.t.\ a coarse-graining function $f\in\mathcal{F}_\mathcal{E}$, then $\bm Z=f(\bm E)$ is a computationally closed coarse-graining of $\bm X$.
    \end{enumerate}
    Now let $\bm E$ additionally be irreducible and aperiodic. Then:
    \begin{enumerate}[resume]
        \item If $\bm Z$ is an informationally closed spatial coarse-graining of $\bm X$, then $\bm E$ is strongly lumpable.
        \item If $\bm E$ is strongly lumpable w.r.t.\ a coarse-graining function $f\in\mathcal{F}_\mathcal{E}$, then, $\bm Z=f(\bm E)$ is an informationally closed coarse-graining of $\bm X$.
    \end{enumerate}
\end{theorem}

\begin{proof}
    We first show that $\bm Z=f(\bm E)$, for an arbitrary spatial coarse-graining $f\in\mathcal{F}_\mathcal{E}$, is a (potentially spatio-temporal) coarse-graining of $\bm
    X$. By
    definition, the $\epsilon$-machine $\bm E$ is a coarse-graining of
    $\bm X$, i.e., there exists an $\epsilon\in\mathcal{G}_{\X}$ such that
    $E_t=\epsilon(\past{X}{t})$. Further, since $Z_t=f(E_t)$ for some
    $f\in\mathcal{F}_{\mathcal{E}}$, $\bm Z$ is a coarse-graining of
    $\bm X$ in the sense that $Z_t=f\big(\epsilon(\past{X}{t})\big)$.

    To prove the first statement, note that computational closure
    implies that there exists a function $f\in\mathcal{F}_\mathcal{E}$
    such that 
    $\bm E'$ (the causal states of the $\epsilon$-machine of $\bm Z$) is given by $E'_t=f(E_t)$. Since, by definition, the causal states of $\epsilon$-machines have Markovian dynamics, weak lumpability follows from this.
    
    For the second statement, assume that $\bm E$ is weakly lumpable
    w.r.t.\ some coarse-graining function
    $f\in\mathcal{F}_\mathcal{E}$. As a consequence $\bm Z$, defined
    by $Z_t=f(E_t)$, is Markov. The $\epsilon$-machine of $\bm Z$ is
    therefore a spatial coarse-graining of $\bm Z$, i.e.,
    $E'_t=\epsilon'(Z_t)$ for some $\epsilon'\in\mathcal{F}_\mathcal{Z}$. We hence
    have that $E'_t=\epsilon'(f(E_t))$, with $\epsilon'\circ
    f\in\mathcal{F}_\mathcal{E}$. Hence, the causal states of the $\epsilon$-machine of
    $\bm Z$ are a (spatial) coarse-graining of the causal states of the $\epsilon$-machine of $\bm X$, which is the
    definition of computational closure. 

    For the third statement, we make use of a fact of $\epsilon$-machines, namely that the current causal state and the future state of the process determine the future causal state, i.e., 
    \begin{subequations}
        \begin{equation}\label{eq:propertyA}
        H(E'_{t+1}|E'_t,Z_{t+1})=0.
    \end{equation}
    Since for informationally closed spatial coarse-grainings we have computational closure, i.e., that $E'_t$ is a function of $E_t$, we have also, by the fact that conditioning reduces entropy,
      \begin{multline}\label{eq:propertyB}
          H(E'_{t+1}|E_t,Z_{t+1})=H(E'_{t+1}|E'_t,E_t,Z_{t+1})\\\le H(E'_{t+1}|E'_t,Z_{t+1})=0 .
      \end{multline}
      and
      \begin{multline}\label{eq:propertyC}
          H(Z_{t+1}|E'_{t+1},E'_t)-H(Z_{t+1}|E'_{t+1},E_t)=\\
          H(Z_{t+1}|E'_{t+1},E'_t)-H(Z_{t+1}|E'_{t+1},E'_t,E_t)=\\
          I(Z_{t+1};E_t|E'_t,E'_{t+1})
      \end{multline}
    \end{subequations}    
    Now note that information closure implies that $I(Z_{t+1};\past{X}{t}|\past{Z}{t})=0$. Hence, with the definitions of $\epsilon$- and $\upsilon$-machines, we get:
    \begin{align}
        0&=H(Z_{t+1}|\past{Z}{t})-H(Z_{t+1}|\past{X}{t})\nonumber\\
        &=H(Z_{t+1}|E'_t)-H(Z_{t+1}|U_t)\nonumber\\
        &\stackrel{(a)}{=}H(Z_{t+1}|E'_t)-H(Z_{t+1}|E_t)\nonumber\\
        &\stackrel{(b)}{=}H(Z_{t+1},E'_{t+1}|E'_t)-H(Z_{t+1},E'_{t+1}|E_t)\nonumber\\
        &\stackrel{(c)}{=} 
        H(E'_{t+1}|E'_t)-H(E'_{t+1}|E_t)+I(Z_{t+1};E_t|E'_t,E'_{t+1})\nonumber\\
        &\stackrel{(d)}{\geq}
        H(E'_{t+1}|E'_t)-H(E'_{t+1}|E_t)\nonumber\\
        &\stackrel{(e)}{\geq} 
        0~,\nonumber
    \end{align}
    where $(a)$ follows because for spatial coarse-grainings $U_t$ is a function of $E_t$, since thus $E_t$ must contain all information about $Z_{t+1}$ that is necessary, and since $U_t$ is minimal, and where $(b)$ follows from adding~\eqref{eq:propertyA} and subtracting~\eqref{eq:propertyB},
    $(c)$ follows from applying the chain rule to the conditional
    entropies and~\eqref{eq:propertyC},
    $(d)$ from the non-negativity of the mutual information, and $(e)$ from the fact that informational closure of spatial coarse-grainings implies computational closure, and hence $H(E'_{t+1}|E_t) = H(E'_{t+1}|E_t,E'_t)$. 
    This derivation proves that $H(E'_{t+1}|E'_t)=H(E'_{t+1}|E_t)$, which is the equivalent condition for strong lumpability~\cite[Th.~9]{GeigerTemmel_kLump}.

    We finally show that $\bm Z$ is informationally closed in the sense of Definition~\ref{def:infoclosure} if $\bm E$ is irreducible, aperiodic, and strongly lumpable w.r.t.\ some $f\in\mathcal{F}_\mathcal{E}$. In this setting, we have $H(Z_t|E_{t-1})=H(Z_t|Z_{t-1})$ for any $t$. Hence, for every finite $L$,
    \begin{align}
        &I(\finfut{Z}{t+1}{L};\past{X}{t}|\past{Z}{t})\\
        &= H(\finfut{Z}{t+1}{L}|\past{Z}{t})-H(\finfut{Z}{t+1}{L}|\past{X}{t},\past{Z}{t})\\
        &= \sum_{\ell=t+1}^{t+L} \Big(H(Z_\ell|\past{Z}{\ell-1})-H(Z_\ell|\past{Z}{\ell-1},\past{X}{t})\Big)\\
        &\stackrel{(a)}{\le} \sum_{\ell=t+1}^{t+L} \Big(H(Z_\ell|Z_{\ell-1})-H(Z_\ell|\past{X}{\ell-1}) \Big)\\
        &\stackrel{(b)}{=} \sum_{\ell=t+1}^{t+L} \Big(H(Z_\ell|Z_{\ell-1})-H(Z_\ell|E_{\ell-1}) \Big)\\
        &=0,
    \end{align}
    where $(a)$ follows because $\bm  Z$ is Markov, because 
    $Z_t=f(\epsilon(\past{X}{t}))$ since $E_t=\epsilon(\past{X}{t})$ is lumpable w.r.t.\ $f\in\mathcal{F}_\mathcal{E}$, 
    and from the fact that conditioning reduces entropy, and $(b)$ because all information of $\past{X}{t}$ relevant for the future of $\bm E$'s coarse-graining must be contained in $E_t$. 
    The rest of the proof follows along the lines of Proposition~\ref{prop:Markov_Lumpability}.
\end{proof}

This result provides another view on the relationship between
information/causal closure and computational closure, which can be
seen now via the relationship between strong and weak lumpability. In
particular, this result highlights information closure as a more
`robust' property. 
In effect, due to its relationship with strong lumpability, this shows that information closure is (mostly) a property of the transition probabilities between causal states and not  their marginal distribution. 
In contrast, computational closure --- as weak lumpability --- depends on the marginal as much as on the transitions. 
This property is, hence, less robust than information closure, as small
deviations from the stationary distribution may break the former but
not the latter.

Additionally, this result is also important because it allows us to use numerical methods that has been developed to empirically discover lumpable partitions of a Markov chain. Indeed, Theorem~\ref{teo:lumpability} shows that those methods can be readily used to empower data-driven discovery of emergent coarse-grainings, requiring only mild conditions on the transitions of the corresponding $\epsilon$-machines. 

The literature of practical methods for discovering lumpable coarse grainings is extensive. For example, an algorithm that constructs the coarsest lumpable coarse-graining of a Markov chain using binary search trees has been presented in Ref.~\cite{DERISAVI2003309}, while an eigenvector-based approach is suggested by~\cite{Barr_Eigenvector,Jacobi_Eigenvector}. Also, an algorithm for the certain class of successively lumpable Markov chains was developed in~\cite{Katehakis_Successively}. 
Since in many practically relevant cases the measurements of $\bm X$ are noisy, a number of algorithms for detecting quasi-lumpable coarse-grainings (i.e., for which the lumpability condition holds approximately) have also been developed. An algorithm based on spectral theory was developed by Jacobi in~\cite{Jacobi_Algo}, while algorithms based on the information-theoretic characterization of lumpability have been proposed by~\cite{GeigerEtAl_OptimalMarkovAggregation}. To prevent these algorithmic approaches to terminate at trivial, e.g., i.i.d. coarse-grainings, annealing procedures have been proposed~\cite{Amjad_GeneralizedMA}~\footnote{Code for this type of approaches can be found at \url{https://github.com/stegsoph/Constrained-Markov-Clustering}.}.

\bibliographystyle{IEEEtran}
\bibliography{main}

\end{document}